\documentclass{article}
\usepackage{amsmath}
\usepackage{amssymb}
\usepackage{amsthm}
\usepackage{dsfont}
\usepackage[nolist]{acronym}
\usepackage{xargs}
\usepackage{ifthen}
\usepackage{algorithm}
\usepackage{algpseudocode}
\usepackage{graphicx}
\usepackage{subfig}
\usepackage[process=auto]{pstool}

\newcommand{\N}{\mathds{N}}
\newcommand{\R}{\mathds{R}}
\newcommand{\F}{\mathds{F}}
\newcommand{\len}[1]{\mathop{\text{len}}\left(#1\right)}
\newcommand{\eps}{\varepsilon}
\newcommand{\set}[1]{\left\{#1\right\}}
\newcommand{\abs}[1]{\left|#1\right|}

\newcommand{\dens}{\text{dens}}
\newcommand{\poly}{\text{poly}}
\newcommand{\LN}{\mathcal{L}_N}
\newcommand{\Y}{\mathcal{Y}}
\newcommand{\X}{\mathcal{X}}
\newcommand{\cc}[1]{\textsc{#1}}
\newcommand{\dtime}{\cc{Dtime}}

\newcommand{\size}{\text{size}}
\renewcommand{\P}{\textsc{P}}
\newcommand{\NP}{\textsc{NP}}
\newcommand{\exptime}{\textsc{Exptime}}

\newcommand{\ceil}[1]{\left\lceil#1\right\rceil}
\newcommand{\floor}[1]{\left\lfloor#1\right\rfloor}

\newcommand{\To}{\Rightarrow}
\newcommand{\TM}{\mathcal{TM}}
\newcommand{\C}{\mathcal{C}_{\poly(\ell)}}
\newcommand{\asgeq}{\geq_{\text{asymp}}}

\newtheorem{thm}{Theorem}[section]{\bfseries}{\rmfamily}
\newtheorem{lem}[thm]{Lemma}{\bfseries}{\rmfamily}
{\bfseries}{\rmfamily}
\newtheorem{cor}[thm]{Corollary}{\bfseries}{\rmfamily}
\newtheorem{defn}[thm]{Definition}{\bfseries}{\rmfamily}
{\bfseries}{\rmfamily}
\newtheorem{rem}[thm]{Remark}{\itshape}{\rmfamily}

\newtheorem{assumption}[thm]{Assumption}{\bfseries}{\rmfamily}
{\bfseries}{\rmfamily}

\title{On the Existence of Weak One-Way Functions}
\author{Stefan Rass\thanks{LIT Secure and Correct Systems Lab, Johannes Kepler University Linz, and Institute for Applied Informatics and Cybersecurity, University of Klagenfurt, email: stefan.rass@jku.at}}

\begin{document}

\maketitle

\begin{abstract}
This note is an attempt to unconditionally prove the existence of weak
\acp{OWF}. Starting from a provably intractable decision problem $L_D$
(whose existence is nonconstructively assured from the well-known discrete
Time Hierarchy Theorem from complexity theory), we construct another
provably intractable decision problem $L\subseteq \set{0,1}^*$ that has its
words scattered across $\set{0,1}^\ell$ at a relative frequency $p(\ell)$,
for which upper and lower bounds can be worked out. The value $p(\ell)$ is
computed from the density of the language within $\set{0,1}^\ell$ divided
by the total word count $2^\ell$. It corresponds to the probability of
retrieving a yes-instance of a decision problem upon a uniformly random
draw from $\set{0,1}^\ell$. The trick to find a language with known bounds
on $p(\ell)$ relies on switching from $L_D$ to $L_0:=L_D\cap L'$, where
$L'$ is an easy-to-decide language with a known density across
$\set{0,1}^*$. In defining $L'$ properly (and upon a suitable G\"odel
numbering), the hardness of deciding $L_D\cap L'$ is inherited from $L_D$,
while its density is controlled by that of $L'$. The lower and upper
approximation of $p(\ell)$ then let us construct an explicit threshold
function (as in random graph theory) that can be used to efficiently and
intentionally sample yes- or no-instances of the decision problem
(language) $L_0$ (however, without any auxiliary information that could
ease the decision like a polynomial witness). In turn, this allows to
construct a weak \ac{OWF} that encodes a bit string $w\in\set{0,1}^*$ by
efficiently (in polynomial time) emitting a sequence of randomly
constructed intractable decision problems, whose answers correspond to the
preimage $w$.
\end{abstract}

\pagebreak

\tableofcontents

\section{Preliminaries and Notation}\label{sec:preliminaries}
Let $\Sigma=\set{0,1}$ be the alphabet over which our strings and encodings
will be defined using regular expression notation. A subset
$L\subseteq\Sigma^*$ is called a \emph{language}. Its complement set (w.r.t.
$\Sigma^*$) is denoted as $\overline L$. The number of bits constituting the
word $w\in\Sigma^*$ is denoted as $\len{w}$, and $w\in\Sigma^*$ can be
explicitly written as a string $w=b_1b_2\ldots b_{\len{w}}$ of bits
$b_i\in\set{0,1}$ (in regular expression notation). The symbol
$(w)_2=\sum_{i=0}^{\len{w}-1}2^i\cdot b_{\len{w}-i}$ is the integer obtained
by treating the word $w\in\set{0,1}^*$ as a binary number, with the
convention of the least significant bit is located at the right end of $w$.

The symbols $\abs{w}$ or $\abs{W}$ will exclusively refer to absolute values
if $w$ is a number (always typeset in lower-case) or cardinality if $W$ is a
set (always written in upper-case)\footnote{We use the symbol $\len{w}$ to
avoid confusion with the word length that is elsewhere in the literature
commonly denoted as $\abs{w}$ too.}. In the following, we assume the reader
to be familiar with \acp{TM} and circuit models of computation. Our
presentation will thus be confined to the minimum of necessary detail, based
on the old yet excellent account of \cite{Hopcroft1979}.

Circuits are here understood as a network of interconnected logical gates,
all of which have a constant maximal number of input signals (bounded
fan-in). For a circuit $C$, we write $\size(C)$ to mean the number of gates
in $C$ (circuit complexity). Formally, the circuit is represented as a
directed acyclic graph, whose nodes are annotated with the specific functions
that they compute (logical connectives, arithmetic operations, etc.). Both,
\acp{TM} and circuits will be designed as decision procedures for a language
$L$; the output is hence a single 1 or 0 bit interpreted as either ``yes'' or
``no'' for the decision problem $w\stackrel ?\in L$ upon the input word $w$.

A complexity class is a set of languages that are decidable within the same
time-limits. Concretely, for a \ac{TM} $M$, let $time_M(w)$ denote the number
of transitions that $M$ takes to halt on input $w$. A language $L$ is said to
be in the complexity class $\dtime(t)$, if a deterministic \ac{TM} exists
that outputs ``yes'' if $w\in L$ or ``no'' if $w\notin L$, on input $w$
within time $time_M(w)\leq t(\len w)$. The language $L(M)$ decided by a
\ac{TM} $M$ is defined as the set of all words $w\in\Sigma^*$ that $M$
accepts by outputting ``yes'' (or any equivalent representation thereof). A
function $f$ is called \emph{fully time-constructible}, if a \ac{TM} $M_f$
exists for which $time_{M_f}(w)=f(\len{w})$ for all words $w\in\Sigma^*$.

Finally, we assume $0\notin\N$ and let all logarithms have base 2.

\section{One-Way Functions}
Our preparatory exposition of \ac{OWF} is based on the account of
\cite[Chp.5]{Zimand2004}. Throughout this work, the symbol $\poly(\ell)$ will
denote different (and not further named) univariate polynomials evaluated at
$\ell$. Throughout this work, and not explicitly mentioned hereafter, we will assume a polynomial $p$ to always satisfy $\lim_{\ell\to\infty}p(\ell)=\infty$. As a reminder, we will write $p\asgeq 0$ to mean that the polynomial $p$ is ``asymptotically larger than'' $0$. We call a function $f:\Sigma^*\to\Sigma^*$ \emph{length regular}, if
$\len{w_1}=\len{w_2}$ implies $\len{f(w_1)}=\len{f(w_2)}$. The function
$f_\ell$ is defined by restricting $f$ to inputs of length $\ell$, i.e.,
$f_\ell:=f|_{\Sigma^{\ell}}$. If $f$ is length regular, then for any
$\ell\in\N$, there is an integer $\ell'\leq\poly(\ell)$ so that
$f_\ell:\Sigma^{\ell}\to\Sigma^{\ell'}$. If the converse relation
$\ell\leq\poly(\ell')$ is also satisfied, then we say that $f$ has
\emph{polynomially related input and output lengths}. This technical
assumption is occasionally also stated as the existence of an integer $k$ for
which $(\len{w})^{1/k}\leq \len{f(w)}\leq (\len{w})^k$. It is required to
preclude trivial and uninteresting cases of one-way functions that would
shrink their input down to exponentially shorter length, so that any
inversion algorithm would not have enough time to expand its input up to the
original size. Polynomially related input and output lengths avoid this
construction, which is neither useful in cryptography nor in complexity
theory \cite{Zimand2004}.

With this preparation, we can state the general definition of one-way
functions, for which we prove non-emptiness in a particular special case
(Definition \ref{def:weak-owf}):

\begin{defn}[one-way function; cf. \cite{Zimand2004}]\label{def:owf}
Let $\varepsilon:\N\to[0,1]$ and $S:\N\to\N$ be two functions that are
considered as parameters. A length regular function $f:\Sigma^*\to\Sigma^*$
with polynomially related input and output lengths is a
\emph{$(\eps,S)$-one-way function}, if both of the following conditions are
met:
\begin{enumerate}
  \item There is a deterministic polynomial-time algorithm $M$ such that,
      for all $w\in\Sigma^*$, $M(w)=f(w)$
  \item For all sufficiently large $\ell$ and for any circuit $C$ with
      $\size(C)\leq S(\ell)$,
      \begin{equation}\label{eqn:owf-inversion-error}
        \Pr_{w\in\Sigma^\ell}\left[C(f_\ell(w))\in f_\ell^{-1}(f_\ell(w))\right]<\eps(\ell)
      \end{equation}
\end{enumerate}
\end{defn}
Observe that Definition \ref{def:owf} does not require $f$ to be a bijection
(we will exploit this degree of freedom later).

In Definition \ref{def:owf}, we can w.l.o.g. replace the deterministic
algorithm to evaluate an \ac{OWF} by a probabilistic such algorithm, upon the
understanding of a probabilistic \ac{TM} as a particular type of
\emph{nondeterministic} \ac{TM} that admits at most two choices per
transition \cite{Sipser2013}. This creates a total of $\leq 2^k$ execution
branches over $k$ steps in time. Assuming a uniformly random bit
$b\in\set{0,1}$ to determine the next configuration (where the transition is
ambiguous), we can equivalently think of the probabilistic \ac{TM} using a
total of $k$ stochastically independent bits (denoted by $\omega$) to define
one particular execution branch $B$, with likelihood
$\Pr_{\omega}[B]=2^{-k}$. In this notation, $\omega\in\set{0,1}^k$ is an
auxiliary string that, for each ambiguous transition, pins down the next
configuration to be taken. So we can think as a probabilistic \ac{TM} to act
\emph{deterministically} on its input word $w$ \emph{and} an \emph{auxiliary
input} $\omega\in\set{0,1}^k$, whose bits are chosen uniformly and
stochastically independent. This view of probabilistic \ac{TM} as
deterministic \ac{TM} with auxiliary input will become important in later
stages of the proof.

For cryptographic purposes, we are specifically interested in strong one-way
functions, which are defined as follows:
\begin{defn}[strong one-way function; cf. \cite{Zimand2004}]\label{def:strong-owf}
A length-regular function $f:\Sigma^*\to\Sigma^*$ with polynomially related
input and output lengths is a \emph{strong one-way function} if for every
polynomial $p\asgeq 0$, $f$ is $(\frac 1{p(\ell)},p(\ell))$-one-way.
\end{defn}
Actually, a much weaker requirement can be imposed, as strong one-way
functions can efficiently be constructed from weak one-way functions (see
\cite[Thm.5.2.1]{Zimand2004} for a proof), defined as:
\begin{defn}[weak one-way function; cf. \cite{Zimand2004}]\label{def:weak-owf}
A length-regular function $f:\Sigma^*\to\Sigma^*$ with polynomially related
input and output lengths is a \emph{weak one-way function} if there is a
polynomial $q\asgeq 0$ such that for any polynomial $p$, $f$ is $(1-\frac
1{q(\ell)},p(\ell))$-one-way.
\end{defn}

Our main result is the following, here stated in its short version:
\begin{thm}\label{thm:weak-owf-exist}
Weak one-way functions exist (unconditionally).
\end{thm}
The rest of the paper is devoted to proving this claim.

\section{Proof Outline and Preparation}\label{sec:outline}
Given a word $w=b_1b_2\ldots b_n\in\set{0,1}^*$, the idea is to map each
1-bit into a yes-instance and each 0-bit into a no-instance of some
intractable decision problem $L_D$. The existence of a suitable language
$L_D$ is assured by the deterministic Time Hierarchy Theorem (Theorem
\ref{thm:time-hierarchy}). If the intended sampling of random yes- and
no-instances can be done in polynomial time, preserving that the decision
problem takes more than polynomial effort (on average), then we would have a
one-way function, illustrated in Figure \ref{fig:owf-construction}.

%

\begin{figure}
\centering
\subfloat[Mapping of $w$ under the \ac{OWF} $f$]{
    \includegraphics[scale=0.7]{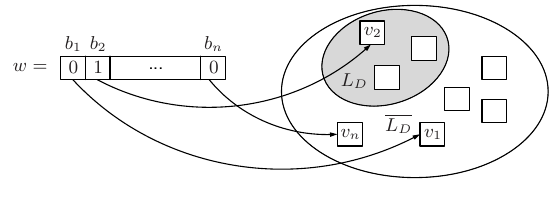}
    }\label{fig:owf-mapping}

    \subfloat[Inversion of the mapping]{
    \includegraphics[scale=0.7]{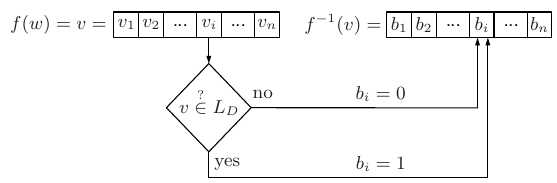}\label{fig:owf-inverse-mapping}}
\caption{\ac{OWF} construction idea}\label{fig:owf-construction}
\end{figure}

The tricky part is of course the sampling, since we cannot plainly draw
random elements and test membership in $L_D$, since this would take more than
polynomial time (by construction of $L_D$). To mitigate this, we change $L_D$
into a language $\LN$ of $N$-element sets of words, redefining the decision
problem as $W\in \LN\iff W\cap L_D\neq \emptyset$. That is, an element $W$ as
being a set of words, is in $\LN$ if and only if at least one of its members
is from $L_D$, but we do not demand any knowledge about which element that
is.

The so-constructed language $\LN$ has the following properties (proven as
Lemma \ref{lem:l-star-hardness}):
\begin{enumerate}
  \item It is at least as difficult to decide as $L_D$, since to decide
      $W\stackrel? \in \LN$, we either have to classify one entry in $W$ as
      being from $L_D$, or otherwise certify that all elements in $W$ are
      outside $L_D$ (equally difficult as deciding $L_D$, since
      deterministic complexity classes are closed under complement).
  \item The property $W\in\LN$ is monotone, in the sense that $W\in\LN$
      implies $V\in\LN$ for all $V\supseteq W$.
\end{enumerate}
The monotony admits the application of a fact that originally rooted in
random graph theory (Theorem \ref{thm:threshold-function}), which informally
says that ``every monotonous property has a threshold''. Intuitively (with a
formal definition of the threshold function being part of the full statement
of Theorem \ref{thm:threshold-function}), a threshold is a function $m$,
whose purpose is most easily explained by resorting to an urn experiment:
consider an urn of $N$ balls in total, $n$ among them being white and $N-n$
balls being black. The threshold depends on $N$ and $p=n/N$, and relates to
drawing from the urn without replacement as follows:
\begin{itemize}
  \item If we draw (asymptotically) less than $m(N,p)$ balls from the urn,
      then the chance to get a white ball asymptotically vanishes as
      $N\to\infty$.
  \item If we draw (asymptotically) more than $m(N,p)$ balls from the urn,
      then the probability to get at least one white ball goes to 1 as
      $N\to\infty$.
\end{itemize}
Now, let us apply this idea to our sampling problem above:
\begin{itemize}
  \item White balls represent yes-instances, i.e., word from $L_D$, and
      black balls represent no-instances, i.e., words from
      $\overline{L_D}$.
  \item The urn is a subset of $\Sigma^*$ of size $N$. To meaningfully
      define such sets with given size, we use a G\"odel numbering of words
      and define our urn to contain $N$ words corresponding to the G\"odel
      numbers $1,2,\ldots,N$. When $m$ denotes the threshold, we can get
      good chances to draw:
      \begin{itemize}
        \item a yes-instance $W$ (with at least one word from $L_D$ in
            it), by taking more than $m$ words,
        \item a no-instance $W$ (having $W\cap L_D=\emptyset$), by taking
            less than $m$ words.
      \end{itemize}
\end{itemize}
The important observation here is that the assurance of having a yes- or
no-instance is given without any explicit testing, yet at the cost of being
only probabilistic. As a technical detail, we need to assure that whether we
have a yes- or no-instance must not become visible by the size of $W$. This
is easily assured by exploiting some sort of relativity: since the threshold
depends on the size of the urn, we can under- or overshoot it by varying the
size of the urn, while leaving the number $\abs{W}$ of elements constant.
This creates equally sized instances $W$ in both cases, with their answer
only determined by the size of the urn; an information that does not show up
in the output of our \ac{OWF}.

Asymptotically, we are almost there, since we already have some useful
properties:
\begin{itemize}
  \item We can sample yes- and no-instances with probability 1
      (asymptotically),
  \item without having to decide $L_D$ or $\LN$ explicitly, and
  \item the sampling could (yet to be verified) run in polynomial time,
      provided that the threshold function behaves properly.
\end{itemize}
So, our next task is working out the threshold function, which depends on the
frequency of words from $L_D$ occurring along the (canonic) enumeration of
$\Sigma^*$ induced by the G\"odel numbering. Alas, the diagonalization
argument that gives us the (initial) language $L_D$ is non-constructive and
in particular gives no clue on how often words from $L_D$ appear in
$\Sigma^*$.

\begin{rem}
Here, in throughout the rest of this work, when we talk about the
``scattering'' of a language $L$, we mean the exact locations of its words on
the line $\N$ of integers. Likewise, the ``density'' of $L$ merely counts the
absolute frequency of words in $L$ up to a certain limit.
\end{rem}

Towards getting an approximate count of words in $L_D$ inside the set of $N$
words with G\"odel numbers $1,2,\ldots,N$, we use a trick: we intersect $L_D$
with a language of known density (formally defined in Section
\ref{sec:density}) that is reducible to $L_D$. Our language of choice
contains all square integers, and defines a new base language $L_0=L_D\cap
SQ$ for $\LN$. This language is at least as hard to decide as $L_D$ (Lemma
\ref{lem:hard-language}) and will replace $L_D$ in the above construction. It
has some important new features:
\begin{itemize}
  \item It gives upper bounds (Lemma \ref{lem:density-upper-bound}) on the
      number of words up to G\"odel number $N$ (trivially, since there
      cannot be more words in $L_D\cap SQ$ than in $SQ$, and the latter
      count is simple).
  \item It also gives lower bounds on the word count, based on a polynomial
      reduction of $SQ$ to $L_D$, illustrated in Figure
      \ref{fig:lower-bound-illustration}.
\end{itemize}
With this, we can complete the sampling procedure along the following steps
(expanded in Section \ref{sec:threshold-sampling}):
\begin{enumerate}
  \item Work out the threshold function explicitly (in fact, we will derive
      upper and lower bounds for it in expression \eqref{eqn:mu-bounds})
  \item Analyze the growth of the threshold function to assure that the
      number of words predicted for the sampling is meaningful (assured by
  \eqref{eqn:growth-of-m} from below, and by \eqref{eqn:mu-bounds} from
  above). Our use of the G\"odel number in connection with the threshold
  bounds lets us choose the urn size polynomial in the length of the input
  word, so that the overall sampling algorithm (sketched next) runs in
  polynomial time (Lemma \ref{lem:threshold-sampling}).
  \item Define the sampling algorithm based on the aforementioned urn
      experiment as follows:
      \begin{itemize}
        \item For a no-instance, make the urn ``large'', such that a
            selection of $\abs{W}=k$ elements will (with probability $\to
            1$) not contain any word from $L_0$ (``white ball'').
        \item For a yes-instance, make the urn ``small'' (relative to
            $k$), such that among $\abs{W}=k$ elements, we have a high
            probability $(\to 1)$ of getting a word from $L_0$.
      \end{itemize}
\end{enumerate}
We call this procedure \emph{threshold sampling}. It allows to realize the
mapping depicted in Figure \ref{fig:owf-construction}, and leaves the mere
task of verifying the properties of an \ac{OWF} according to Definition
\ref{def:weak-owf}. The evaluation of the function in polynomial time means
to repeatedly sample, each run taking polynomial time (in the number of
bits). This will directly become visible in the construction and from the
properties of the threshold.

Showing the intractability of inversion is trickier, but here we can make use
of the fact that Definition \ref{def:owf} does not require the function to be
bijective. In fact, the use of randomness in the sampling necessarily renders
our constructed mapping not bijective, but any inversion algorithm working on
the image $y=f(b_1b_2\ldots b_n)$ would necessarily also return a correct
first bit $b_1$. Taking a contraposition, this means that the chances for the
inversion to fail are at least those to screw up the computation of $b_1$
(the argument is expanded in full detail in Section
\ref{sec:weak-owf-exist}). But this is exactly how the diagonal language
$L_D$ was constructed for, in the worst case. So our last challenge is making
the worst-case appear with the desired frequency of $1-1/\poly(n)$, as
required for a weak \ac{OWF}. This is done by modifying the encoding of
Turing machines to use only a logarithmically small fraction of its input, so
as to consider a large number of inputs of the same length as equivalent (in
Section \ref{sec:p-vs-np}, we will relate this to the notion of local
checkability \cite{Arora.2007}). This (wasteful) encoding is consistent with
all relevant definitions (especially Definition \ref{def:owf}), but makes the
worst-case occur with a non-negligible frequency (as we require).

Having outlined all ingredients, let us now turn to the formal details,
starting with some preparation.

\subsection{G\"odel Numbering}\label{sec:goedel-and-density}
To meaningfully associate subsets $\set{1,2,\ldots,N}\subset\N$ with subsets
of $\Sigma^*$, let us briefly recall the concept of a G\"odel numbering. This
is a mapping $gn:\Sigma^*\to\N$ that is computable, injective, and such that
$gn(\Sigma^*)$ is decidable and $gn^{-1}(n)$ is computable for all $n\in\N$
\cite{Hermes1971}. The simple choice of $gn(w)=(w)_2$ is obviously not
injective (since $(0^nw)_2=(w)_2$ for all $n\in\N$ and all $w\in\Sigma^*$),
but this can be fixed conditional on $0\notin\N$ by setting
\begin{equation}\label{eqn:goedelnr}
gn(w) := (1w)_2.
\end{equation}
This is the G\"odel numbering that we will use throughout the rest of this
work, and it is not difficult to verify the desired properties as stated
above. Most importantly, \eqref{eqn:goedelnr} is a computable bijection
between $\N$ and $\Sigma^*$.

For the G\"odelization of \acp{TM}, let $\rho(M)\in\Sigma^*$ denote a
complete description of a \ac{TM} $M$ in string form (using some prefix-free
encoding to denote the alphabet, state transitions, etc.). The encoding that
we will use (and define in Section \ref{sec:tm-encoding}) will have the
following properties (as are commonly required; cf.
\cite{Arora2009,Hopcroft1979}):
\begin{enumerate}
  \item every string over $\set{0,1}^*$ represents \emph{some} \ac{TM}
      (easy to assure by executing an invalid code as a canonic \ac{TM}
      that instantly halts and rejects its input),
  \item every \ac{TM} is represented by infinitely many strings. This is
      easy by introducing the convention to ignore a prefix of the form
      $1^*0$ then the string representation is being executed.
\end{enumerate}
The G\"odelization of a \ac{TM} $M$, represented as $\rho(M)\in\Sigma^*$, is
then the integer $gn(\rho(M))$.

\subsection{Density Functions}\label{sec:density}
For a language $L$, we define its \emph{density function}, w.r.t. a G\"odel
numbering $gn$, as the mapping
\[
\dens_L: \N\to\N,\quad x\mapsto\abs{\set{w\in L: gn(w)\leq x}},
\]
i.e., $\dens_L(x)$ is the number of words whose G\"{o}del
number\footnote{Other definitions of the density \cite{Papadimitriou1994},
differ here by counting words up to a maximal length. This would be too
coarse for our purposes.} as defined by \eqref{eqn:goedelnr} is bounded by
$x$. The dependence of $\dens_L$ on the G\"odel numbering $gn$ can be omitted
hereafter, since there will be no second such numbering and hence no
ambiguity by this simplification of the notation.

Occasionally, it will be convenient to let $\dens_L$ send a word $v\in\Sigma^*$
to an integer $\N$, in which case we put $x:=gn(v)=(1v)_2$ in the definition of
$\dens_L$ upon an input word $v$. The density of the language $L$ will be our
technical vehicle to quantify (bound) the likelihood of drawing an element
from $L$ within a bounded set of integers $\set{1,2,\ldots,n}$ (see Lemma
\ref{eqn:density-bounded-by-argument} below), where the bound $n$ will be an
integer or a binary number coming as a string (whichever is the case will be
clear from the context).

\section{Proof of Theorem \ref{thm:weak-owf-exist}}\label{sec:proof}
The proof will cook up a weak \ac{OWF} from the ingredients outlined in
Section \ref{sec:outline}, in almost bottom up order.

\subsection{Properties of Density Functions}\label{sec:density-functions}
Our first subgoal is the ability to construct random yes- and no-instances of
a difficult decision problem. So, we first need to relate the density
function for a language $L$ to the likelihood of retrieving elements from it
upon uniformly random draws.

\begin{lem}\label{eqn:density-bounded-by-argument}
For every language $L$, the density function satisfies $\dens_L(x)\leq x$ for
all $x\in\N$.
\end{lem}
\begin{proof}
Assume the opposite, i.e., the existence of some $x_0$ for which
$\dens_L(x_0)>x_0$. In that case, there must be at least $x_0+1$ words
$w_1,w_2,\ldots,w_{x_0+1}$ in $L$ for which $gn(w_i)\leq x_0$ for all
$i=1,2,\ldots,x_0+1$. W.l.o.g., let $w_1$ be the word whose G\"odel number
$gn(w_1)$ is maximal. Since $gn$ is injective, all other $x_0$ words map to
distinct smaller integers, thus making $gn(w_1)\geq x_0+1$ at least. This
clearly contradicts our assumption that $gn(w_1)\leq x_0$.
\end{proof}
Lemma \ref{eqn:density-bounded-by-argument} permits the use of the density
function to define an urn experiment as follows: let the urn be
$U=\set{1,2,\ldots,n}\subset\N$, and let each element in it correspond to a
word $w\in\Sigma^*$ by virtue of $gn^{-1}$. Then the likelihood to draw an
element from $L$ addressed by a random index in $U$ is $\dens_L(n)/n$, by
counting the number of positive cases relative to all cases.

To illustrate the practical use of a density function, let us consider the
following example of a language that we will heavily use throughout this 
work. The language of \emph{integer squares} is defined as $SQ=\{y: \exists x\in\N$ such that $y=x^2\}$.
Each element $y\in SQ$ can be identified with a string (in regular
expression notation) $w_y\in 1(0\cup 1)^*\subset\Sigma^*$, for which $y=(w_y)_2$.
The G\"odel number of $w_y$ can be computed from $y$ by
$gn(w_y)=2^{\ceil{\log y}+c(y)}+y$, with the padding function
\[
    c(y) = \left\{
             \begin{array}{ll}
               0, & \hbox{if }\log y<\ceil{\log y}; \\
               1, & \hbox{if }\log y=\ceil{\log y}.
             \end{array}
           \right.
\]

Let us extend our definition of $gn$ to a mapping from $\N\to\N$, where
$gn(y)$ for $y\in\N$ is defined as $gn(y):=gn(w_y)$ with $y=(w_y)_2$. Using
the previous formula to compute $gn(y)$, note that the expression
\[
    \frac{gn(y)}y=\frac{2^{\ceil{\log y}+c(y)}+y}y = 1 + \frac{2^{\ceil{\log y}+c(y)}}y,
\]
ultimately becomes numerically trapped within the interval $[1,5]$ for
$y\to\infty$ (the lower bound is immediate; the upper bound follows from
$2^{\ceil{\log y}+c(y)}\leq 2^{1+(\log y)+ 1}=4y$). Thus,
\begin{equation}\label{eqn:square-gn-bounds}
  y\leq gn(y)\leq 5\cdot y\quad\text{for sufficiently large}~y.
\end{equation}
Moreover, it is easy to see that for $z,x\in\N$,
\begin{equation}\label{eqn:square-set-upper-bound}
\abs{\set{z^2: z^2\leq x}}=\floor{\sqrt x}.
\end{equation}

Using both facts, we discover that for any two $x,z\in\N$ that satisfy
$gn(z^2)\leq x$, also $z^2\leq gn(z^2)\leq x$ holds by
\eqref{eqn:square-gn-bounds}. Thus, $[gn(z^2)\leq x]\To[z^2\leq x]$ and hence
$\set{z^2: gn(z^2)\leq x}\subseteq\set{z^2:z^2\leq x}$. The cardinalities of
these sets satisfy the respective inequality, and
\eqref{eqn:square-set-upper-bound} gives
\begin{equation}\label{eqn:sq-upper-bound}
\dens_{SQ}(x)\leq\floor{\sqrt{x}}\leq\sqrt{x}.
\end{equation}

Conversely, $gn(z^2)\leq 5z^2$ asymptotically by \eqref{eqn:square-gn-bounds}
means that for sufficiently large $z$, $gn(z^2)\leq 5\cdot z^2\iff \frac 1
5\cdot gn(z^2)\leq z^2$. Thus, $[z^2\leq x]\To[\frac 1 5\cdot gn(z^2)\leq
x]$, and the last condition is equivalent to $gn(z^2)\leq 5\cdot x$.
Therefore, $\set{z^2:z^2\leq x}\subseteq\set{z^2:gn(z^2)\leq 5\cdot x}$, and
the cardinalities satisfy the respective inequality. It follows that
$\dens_{SQ}(5\cdot x)\in\Omega(\sqrt x)$, or after substituting and renaming
the variables, $\dens_{SQ}(x)\in\Omega(\sqrt x)$.

Summarizing our findings, we have proven:
\begin{lem}\label{lem:density-of-squares}
The language of squares $SQ=\set{y: y=x^2, x\in\N}$ has a density function
$\dens_{SQ}(x)\in\Theta(\sqrt x)$.
\end{lem}

As announced in Section \ref{sec:outline}, we will later look at the density
of the intersection of two languages (namely $L_D\cap SQ$, where $L_D$ has
not been constructed explicitly yet). The definition of density functions
immediately delivers a useful inequality for such intersection sets: for
every two languages $L_1,L_2$, we have
\begin{equation}\label{eqn:density-of-intersections}
  \dens_{L_1\cap L_2}\leq \dens_{L_1},
\end{equation}
since there cannot be more words in $L_1\cap L_2$ than words in $L_1$ (or
$L_2$, respectively). This will enable us to bound the density of the (more
complex) intersection language in terms of the simpler (and known) density of
$SQ$. Details will be postponed until a little later.

\subsection{Encoding of Turing Machines}\label{sec:tm-encoding}
As a purely technical matter, we will adopt a specific encoding convention
for \acp{TM}. While the following facts are almost trivial, it is important
to establish them a-priori (and thus independently) of our upcoming
arguments, since the scattering and density of the languages that we
construct will depend on the chosen encoding scheme of \acp{TM}.
Specifically, we will encode a \ac{TM} $M$ into a string $\rho(M)$ as
outlined in Section \ref{sec:goedel-and-density}, with a few adaptations when
it comes to executing a code for a \ac{TM}:

\begin{itemize}
  \item When a \ac{TM} as specified by an input $w\in\Sigma^*$ is to be
      executed by a universal \ac{TM} $M_U$, then the code $\rho(M)$ that
      defines $M$'s actions is obtained by $M_U$ as follows:
      \begin{itemize}
        \item the input $w$ is treated as an integer $x=(w)_2$ in binary
            and all but the most significant $\ceil{\log(\len{w})}$ bits
            are ignored. Call the resulting word $w'$.
        \item from $w'$, we drop all preceding $1$-bits and the first
            0-bit, i.e., if $w'=1^k0v$, then $\rho(M)=v$ after discarding
            the prefix padding $1\ldots 10$.
      \end{itemize}
\end{itemize}

\begin{figure}[bh!]
  \centering
  \includegraphics[scale=0.8]{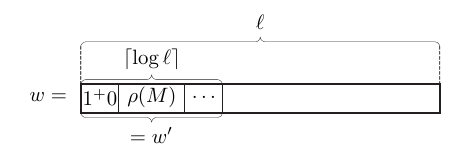}
  \caption{Encoding of Turing machines with padding}\label{fig:encoding}
\end{figure}

Although this encoding -- depicted in Figure \ref{fig:encoding} -- is
incredibly wasteful (as the code for a \ac{TM} is taken as padded with an
exponential lot of bits), it assures several properties that will become
useful at the beginning and near the end of this work:
\begin{enumerate}
\item The aforementioned mapping $w \mapsto w'$ shrinks the entirety of
    $2^{\ell}$ words in $\set{0,1}^\ell$ down to only
    $2^{\ceil{\log\ell}}\geq \ell$ distinct prefixes. Each of these admits
    a lot of $2^{\ell-\ceil{\log\ell}}$ suffixes that are irrelevant for
    the encoding of the \ac{TM}. Thus, an arbitrary word $w'$ encoding a
    \ac{TM} has at least
    \begin{equation}\label{eqn:equivalent-encodings}
    2^{\ell-\ceil{\log\ell}}\geq 2^{\ell-\log\ell-1}
    \end{equation}
    equivalents $w$ in the set $\set{0,1}^\ell$ that map to $w'$. Thus, if
    a \ac{TM} $M$ is encoded within $\ell$ bits, then
    \eqref{eqn:equivalent-encodings} counts how many equivalent codes for
    $M$ are found at least in $\set{0,1}^\ell$. This will be used in the
    concluding Section \ref{sec:weak-owf-exist}, when we establish failure
    of any inversion circuit in a polynomial number of cases (second part
    of Definition \ref{def:owf}).

\item The retraction of preceding $1$-bits creates the needed infinitude of
    equivalent encodings of \emph{every} possible \ac{TM} $M$, as we can
    embed any code $\rho(M)$ in a word of length $\ell$ for which
    $\log(\ell)>\len{\rho(M)}$. We will need this to prove the hierarchy
    theorem in Section \ref{sec:time-hierarchy-theorem}.
\end{enumerate}

\begin{rem}\label{rem:exponential-growth}
Note that exponential difference in the size of $\rho(M)$ and its
representation $w$ in fact does \emph{not} preclude the efficient execution
of $w$ as input code and data to the universal \ac{TM}, because it only
executes a logarithmically small fraction of its input code. Conversely, the
redundancy of our encoding only means that we have to reach out exponentially
far on $\N$ to see the first occurrence of a \ac{TM} with a code of given
size; this is, however, not forbidden by any of the relevant definitions.
\end{rem}

Let $\TM=\set{M_1,M_2,M_3,\ldots}$ be an enumeration of all \acp{TM} under
the encoding just described; that is, $\TM$ is the set of all $w\in\Sigma^*$
for which a \ac{TM} $M$ with encoding $\rho(M)$ exists that is embedded
inside $w$ as shown in Figure \ref{fig:encoding}. Observe that the first
$1$-bit (mandatory in our encoding) when being stripped from a word $w$ by
$gn^{-1}$ leaves the inner representation of $M$ intact (since the
$1^+0$-prefix is ignored for the ``execution'' of $w$ anyway). We write $M_w$
to mean the \ac{TM} encoded by $w$.

A simulation by the universal \ac{TM} $M_U$ thus takes the program $w$ and
input $x$ to act on the initial tape configuration $\#w\#x\underline{\#}$, or
in expanded form (cf. Figure \ref{fig:encoding}),
\begin{equation}\label{eqn:tape-config}
  \#\hspace{-2.7mm}\underbrace{1^+0}_{\substack{\text{padding}\\\text{(ignored)}}}\hspace{-2.5mm}\overbrace{\rho(M)}^{\text{code}}\underbrace{(0\cup 1)^*}_{\substack{\text{padding}\\\text{(ignored)}}}\#x\underline{\#}
\end{equation}
where $\#$ marks spaces on the tape, and the head position is marked by the
underlining.

\subsection{A Review of the Time Hierarchy Theorem}\label{sec:time-hierarchy-theorem}

Returning to the proof outline, our next goal is to find a proper difficult
language $L_D$ that we can use for the encoding of input bits into yes/no
instances of a decision problem. To this end, it is useful to take a close
look at the proof of the deterministic time hierarchy theorem known from
complexity theory. The theorem's hypothesis is summarized as follows:

\begin{assumption}\label{asm:time-hierarchy}
Let $T:\N\to\N$ be a fully time-constructible function, and let $t:\N\to\N$ with $t(n)\geq n$ be a
monotonously increasing function for which
\[
    \lim_{\ell\to\infty}\frac{t(\ell)\cdot\log t(\ell)}{T(\ell)}=0.
\]
\end{assumption}

Theorem \ref{thm:time-hierarchy} is obtained by diagonalization
\cite[Thm.12.9]{Hopcroft1979}: we construct a \ac{TM} $M$ that halts within
no more than $T(\len{w})$ steps upon input of a word $w$ of length $\ell$,
and differs in its output from every other \ac{TM} $M'$ that is
$t(\ell)$-time-limited.

On input of a word $w$ of length $\ell=\len w$, the sought \ac{TM} $M$ will
employ a universal \ac{TM} $M_U$ to simulate an execution of $M_w$ on input
$w$. The simulation of $t(\ell)$ steps of $M_w$ can be done by $M$ taking no
more than $c_{M_w}\cdot t(\ell)\log t(\ell)$ steps
\cite[Thm.12.6]{Hopcroft1979}, where $c_{M_w}$ is a constant that depends
only on the number of states, tapes, and tape-symbols that $M_w$ uses, but
not the length of the input to ($M_U$'s simulation of) $M_w$.

To assure that $M$ always halts within the limit $T(\ell)$, it simultaneously
executes a ``stopwatch'' \ac{TM} $M_T$ on the input $w$, which exists since
$T$ is fully time-constructible. Once $M_T$ has finished, $M$ terminates the
simulation of $M_w$ too, and outputs ``accept'' if and only if two conditions
are met:
\begin{enumerate}
  \item $M_w$ halted (by itself) during the simulation (i.e., it was not
      interrupted by the termination of $M_T$), and,
  \item $M_w$ rejected $w$.
\end{enumerate}
The ``diagonal-language'' $L_D$ is thus defined over the alphabet
$\Sigma=\set{0,1}$ as
\begin{equation}\label{eqn:diagonal-language}
    L_D := \set{w\in\Sigma^*: M_w\text{ halts and rejects }\rho(M_w)\text{ within}\leq T(\len{\rho(M_w)})\text{ steps}}.
\end{equation}
\begin{rem}
	Textbook proofs of the time hierarchy theorem, e.g., \cite{Hopcroft1979}, adopt a slightly simpler version of $L_D$, usually a word $w$ entirely be interpreted as some code for a \ac{TM} $M_w$, and having this \ac{TM} process $w$ within time $T(\len{w})$. In foresight of our intention to ``pad'' words into becoming perfect squares (to lie in a (modified version of) $SQ$), this padding would change $w$ into some different word $w'$, but leave the ``functional prefix'' $\rho(M)$ (see Fig. \ref{fig:encoding}) inside $w$ unchanged. Hence, $M_w$ would not simulate its own code, but a modified version thereof. To recover the arguments for the textbook proof of the hierarchy theorem, we restrict the decision to processing only that part of $w$ that contains the \ac{TM} encoding, i.e., $w'$ in Fig. \ref{fig:encoding}. Since we still retain the infinitude of equivalent encodings by the prefix $1^+0$ in Fig. \ref{fig:encoding}, the proof arguments from the textbook \cite{Hopcroft1979} remain intact.
\end{rem}

The hierarchy theorem is then found by observing that $L_D$ cannot be
accepted by any $t$-time-limited \ac{TM} $M$: If $M$ were $t$-time-limited
with encoding $\rho(M)=w'$, then the list $\TM$ contains another (equivalent)
encoding $w$ of length $\ell=\len{w}$ so that $M=M_{w'}$ and $M_w$ compute
identical functions, and for sufficiently large $\ell$,
\begin{equation}\label{eqn:time-bound}
    c_{M_w}\cdot t(\ell)\cdot\log t(\ell)\leq T(\ell),
\end{equation}
so that $M_w$ can carry to completion within the time limit $T(\ell)$. Now,
$w\in L(M_w)$ if and only if $w\notin L_D$, so that $L_D\neq L(M_w)$. Since $M$ was
$t$-time-limited and arbitrary, and $M_w$ decides the same language as $M$,
we have $L_D\neq L(M)$ for all $M$ that are $t$-time-limited, and therefore
$\dtime(t)\subsetneq\dtime(T)$ if also $t\leq T$.

At this point, we just re-proved the following well-known result:
\begin{thm}[deterministic time hierarchy
theorem]\label{thm:time-hierarchy} Let $t,T$ be as in Assumption
\ref{asm:time-hierarchy} and $t\leq T$, then $\dtime(t)\subsetneq\dtime(T)$.
\end{thm}

\subsection{A Hard Language with a Known Density
Bound}\label{sec:hard-languages} The existence of a language $L_D$ that is
hard to decide allows the construction of another language whose scattering
over $\Sigma^*$ can be quantified explicitly. We will intersect $L_D$ with
another language with known density estimates, and show that the hardness of
the implied decision problem is retained. Our language of choice will be
already known set of integer squares that we will (equivalently) redefine for
that purpose to be $SQ:=\{w\in\Sigma^*: \exists x\in\N$ such that $gn(w)=x^2\}$. 
This
language has a density $\dens_{SQ}(x)\in\Theta(\sqrt x)$ by Lemma
\ref{lem:density-of-squares}. 

We claim that the language
\[
L_0 := L_D\cap SQ
\]
is at least as difficult to decide as $L_D$. Assume the opposite
$L_0\in\dtime(t)$ towards a contradiction, and let a word $w\in\Sigma^{\ell}$
be given. Without loss of generality, let us assume that the lower order bits in $w$ are all zero, since the relevant ``functional'' part is the header $1^+\rho(M)$ (cf. \eqref{eqn:tape-config}).

We look for the smallest $w'\geq w$ that approximates $w$ from above and
represents a square number in binary, which is
$(w')_2=\ceil{\sqrt{(w)_2}}^2\geq (w)_2$. Observe that two adjacent integer
squares $x^2$ and $(x+1)^2$ are separated by no more than $(x+1)^2 - x^2 =
2x+1$. Therefore, putting $x=\ceil{\sqrt{(w)_2}}$, we find that the
difference $\Delta$ between $w$ and its upper square approximation $w'$
satisfies $(w')_2 - (w)_2 = \Delta \leq 2\ceil{\sqrt{(w)_2}}+1$. Taking
logarithms to get the bitlength, we find that $\Delta$ takes no more than
$\ceil{\log(2\ceil{\sqrt{(w)_2}}+1)}\leq 3+\frac 1 2\ceil{\log
(w)_2}= 3+\frac{\len w}2$ bits, assuming that $w$ has no leading zeroes (which our G\"odel numbering precludes).

By adding $\Delta$ to $(w)_2$ to get the sought square $(w')_2 =
(w)_2+\Delta$, note that the shorter bitlength of $\Delta$ relative to the
bitlength of $w$ makes $w$ and $w'$ different in the lower half + 4 bits
(including the carry from the addition of $\Delta$). Equivalently, $w$ and
$w'$ have a Hamming distance $\leq \frac 1 2\len{w}+4$. 

Since $\ell-\log(\ell)>4+\frac 1 2\ell$ for sufficiently large $\ell$,
we conclude that $w$ and its ``square approximation'' $w'$ will eventually
have an identical lot of $\ceil{\log\ell}$ most significant bits (cf. Figure
\ref{fig:encoding}). That is, the header of the word that is relevant for
$L_D$ is not touched when $w$ is converted into a square $w'$. This means
that $w\in L_D\iff w'\in L_D$, so that the decision remains unchanged upon
the switch from $w$ to $w'$. Since $w'\in SQ$ holds by construction, we could
decide $w\in L_D$ by deciding whether $w'\in L_D\iff w'\in L_D\cap SQ$, so
that $L_D\in\dtime(t)$ by our initial hypothesis on $L_0$ and the additional assumption that $t(n)\geq n^3$. This contradiction
puts $L_0\notin\dtime(t)$, as claimed. To retain $L_D\cap SQ\in\dtime(T)$, we
must choose $T$ so large that the decision $w\in SQ$ is possible within the
time limit incurred by $T$, so we add $T(n)\geq n^3$ to our hypothesis
besides Assumption \ref{asm:time-hierarchy} (note that we do not need an
optimal complexity bound here).

Using \eqref{eqn:density-of-intersections} with $L_1=L_D$ and $L_2=SQ$, we
see that for sufficiently large $x$,
\[
\dens_{L_0}(x)=\dens_{L_D\cap SQ}(x)\leq\dens_{SQ}(x)\leq \sqrt x,
\]
by \eqref{eqn:sq-upper-bound}. This proves half of what we need, so let us
capture this intermediate finding in a rememberable form:
\begin{lem}\label{lem:density-upper-bound}
Let $t,T$ be as in Assumption \ref{asm:time-hierarchy} and assume $T(n)\geq t(n)\geq 
n^3$. Then, there exists a language $L_0\in\dtime(T)\setminus\dtime(t)$ for
which
\[
    \dens_{L_0}(x)\leq \sqrt x.
\]
\end{lem}

Towards a lower bound for the density, the following observation will turn
out as a key tool:
\begin{lem}\label{lem:dtime-t-hardness}
The language $L_0$ described in Lemma \ref{lem:density-upper-bound} is
$\dtime(t)$-hard (via polynomial reduction).
\end{lem}
\begin{proof}
We need to show that for every $L\in\dtime(t)$, there exists a poly-time
reduction $\varphi$ to the language $L_0$. Remember that by definition
\eqref{eqn:diagonal-language}, $L_D$ is the set of all words $w$ that when
being interpreted as an encoding of a Turing machine $M_w$, this machine
would reject ``itself'' as input within time $T(\len w)$.

Take any $L\in\dtime(t)$, then there is a \ac{TM} $M_L$ that decides
$w\stackrel ? \in L$ in time $t(\len w)$. Let $M_{\overline{L}}$ be the
\ac{TM} that decides $\overline{L}$ (i.e., by simply inverting the answer of
$M_L$). To construct a proper member of $L_D\cap SQ$ that equivalently
delivers this answer, we define the reduction $\varphi(w) = w' =
\rho(M_{\overline{L}})\$w\$0^{\nu}$ for an integer $\nu$ that is specified
later. That is, the word $w'$ contains a description of $M_{\overline{L}}$,
followed by the original input $w$ and a number $\nu$ of trailing zeroes that
will later be used to cast this word into a square. The three blocks in
$\varphi(w)$ are separated by \$-symbols, assuming that \$ is not used in any
of the relevant tape alphabets, and we use a prefix-free encoding.

Let us collect a few useful observations about the mapping $\varphi$:
\begin{itemize}
  \item $\varphi(w)$ is poly-time computable when $\nu=\poly(\len w)$,
      since $\rho(M_{\overline{L}})$ is merely a constant prefix being
      attached. It is especially crucial to remark here that the
      exponential expansion of a \ac{TM} of length $\ell$ into an encoding
      of size $O(2^\ell)$ (cf. Remark \ref{rem:exponential-growth}) does
      not make the complexity to evaluate $\varphi$ exponential, since the
      universal \ac{TM} $M_U$ merely drops padding from the code, but not
      from the entire input word. Indeed, the (padded) code
      $1^+0\rho(M_{\overline{L}})1^*(0\cup 1)^*$ appearing on the \ac{TM}'s
      tape (see \eqref{eqn:tape-config}) is exponentially longer than the
      ``pure'' code for $M_{\overline{L}}$, but it is nevertheless a
      \emph{constant} prefix used by the reduction $\varphi$, since it is
      constructed explicitly for the fixed language $L$. As such, the
      reduction is doable in $O(1)$ time.

      A slight difficulty arises from the need to make $\varphi(w)=w'$
      sufficiently long to give the simulation of $M_{w'}$ enough time to
      finish. This is resolved by increasing $\nu$ (thus making the
      zero-trailer $0^\nu$ longer), so as to enlarge $w'$ until condition
      \eqref{eqn:time-bound} is satisfied. Note that the increase of $\nu$
      depends on $t$ and $T$ only and is as such a fixed number (constant),
      adding to the remainder length of $\nu$ that polynomially depends on
      the length of $w\in L$ only.
  \item The output length $\len{\varphi(w)}$ is again polynomial in $\len
      w$ under the condition that $\nu=\poly(\len w)$.
  \item $\varphi$ is injective, since $w_1\neq w_2$ implies
      $\varphi(w_1)\neq \varphi(w_2)$ by definition.
\end{itemize}

To see why $w\in L\iff \varphi(w)=w'\in L_D$, let us agree on the convention
that the \ac{TM} $M_{\varphi(w)}$ executes $M_{\overline{L}}$ only on that
part of $w'$ that is enclosed within \$-symbols. Leaving our universal
\ac{TM} unmodified, this restriction can be implemented by a proper
modification of $M_{\overline{L}}$ to ignore everything before and after the
\$-symbols during its execution (thus slightly changing the definition of our
reduction to respect this). Let us call the so-modified \ac{TM}
$M'_{\overline{L}}$, and alter the reduction into $\varphi(w) :=
\rho(M'_{\overline{L}})\$w\$0^{\nu}$.

Under these modifications, it is immediate that:
\begin{enumerate}
  \item the simulation of $M_{w'}$ on input $w'$ is actually a simulation of
      $M_{\overline{L}}$ on input $w$, and has -- by construction (a
      suitably large padding of $\nu$ trailing bits) -- enough time to
      finish, and,
  \item the \ac{TM} deciding $L_D$ will accept $w'$ if and only if
      $M_{\overline{L}}$ rejects $w$. In that case, however, $M_L$ would
      have accepted $w$, thus $w\in L\iff \varphi(w)\in L_D$.
\end{enumerate}

It remains to modify our reduction a last time to assure that $\varphi(w)\in
SQ$ for \emph{every} possible $w$, so as to complete the reduction $L\leq_p
L_0$. For that matter, we will utilize the previously introduced trailer of
zeroes $0^{\nu}$ in $\varphi(w)$.

Define the number $k := \len{\rho(M'_{\overline L})\$w\$}=c+\len w$, where
$c$ is a constant that counts the length of $\rho(M'_{\overline L})$ and the
\$-symbols when everything is encoded in binary. We will enforce
$\varphi(w)\in SQ$ by interpreting $w'=\rho(M'_{\overline{L}})\$w\$0^{\nu}$
as a binary number with $\nu$ trailing zeroes, and add a proper value to it
so as to cast $w'$ into the form $(w')_2=x^2$ for some integer $x$. The
argument is exactly as in the proof of Lemma \ref{lem:density-upper-bound},
and thus not repeated but visualized in Figure \ref{fig:lagged-square}.

\begin{figure}
  \centering
    \includegraphics[scale=0.8]{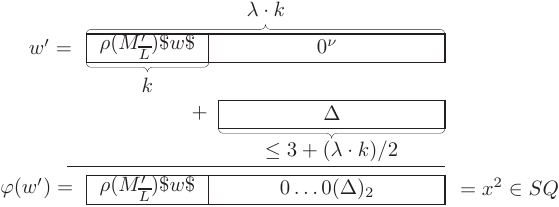}
  \caption{Mapping into the set of squares}\label{fig:lagged-square}
\end{figure}

Now, let $\nu$ be such that $\len {w'}=\lambda\cdot k$ for some (sufficiently
large) integer multiple $\lambda\geq 3$ (see Figure \ref{fig:lagged-square}
to see how $\nu, \len {w'}, k$ and $\lambda$ are related). To cast
$\varphi(w)$ into the sought form $(\varphi(w))_2=x^2$ for an integer
$x\in\N$ (and hence $\varphi(w)\in SQ$), we need to add some $\Delta$ towards
the closest larger integer square. If we choose $\lambda$ so large that
$3+\frac\lambda 2\cdot k<(\lambda-1)k$, then $\nu \geq (\lambda-1)k$
zero-bits (plus the additional lot to satisfy condition
\eqref{eqn:time-bound} if necessary, but for sufficiently long words, this
requirement vanishes) at the end of $w'$ suffice to take up all bits of
$\Delta$, and $(\rho(M'_{\overline L})\$w\$0^\nu)_2+ \Delta
=(\rho(M'_{\overline L})\$w\$0^*z)_2$ (with $z$ being the binary
representation of $\Delta$) is a square. Since $\lambda$ can be chosen as a
fixed integer multiplier for $k=c+\len{w}$, the above requirement
$\nu=\poly(\len w)$ is satisfied, and $\varphi(w)\in SQ$ holds for
\emph{every} input word $w$.

Therefore, $\forall w: \varphi(w)\in SQ$ implies $w\in L\iff \varphi(w)\in L_D\cap
SQ\iff \varphi(w)\in L_D$, and the result follows since $L\in\dtime(t)$ was arbitrary.
\end{proof}

\begin{figure}
  \centering
    \includegraphics[scale=0.8]{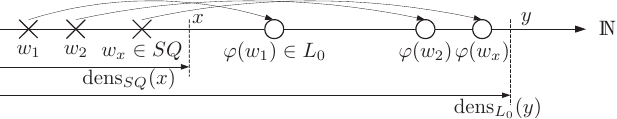}
  \caption{Illustration of inequality \eqref{eqn:lower-bound-1}}\label{fig:lower-bound-illustration}
\end{figure}

Lemma \ref{lem:dtime-t-hardness} lets us lower bound the number of words in
$L_0$ by using any known lower bound for any language in $\dtime(t)$, and
knowing that all these words map into $L_0$ (see Figure
\ref{fig:lower-bound-illustration}). Our language of choice will be $SQ$ once
again, with Lemma \ref{lem:density-of-squares} providing the necessary
bounds. This is admissible if we add the hypothesis $t(n)\geq n^3$ so that
$SQ\in\dtime(t)$. Furthermore, Assumption \ref{asm:time-hierarchy} then
implies that $T(n)\in\Omega(n^3)$ as well, so that this requirement in Lemma
\ref{lem:density-upper-bound} becomes redundant under our so-extended
hypothesis.

We consider the length of a word $w$ being mapped to a word $\varphi(w)\in
L_0=L_D\cap SQ$. For $x\in\N$, let $w_x$ be the last word to appear before
$x$ in an ascending $\leq$-ordering of $SQ$ (see Figure
\ref{fig:lower-bound-illustration}). The mapping $\varphi$ is strictly
increasing in the following sense: the images of two words $w_1\leq w_2$
under $\varphi$ would contain $w_1, w_2$ as ``middle'' blocks in the
bitstrings $\varphi(w_1), \varphi(w_2)$, where they determine the order
$gn(\varphi(w_1))\leq gn(\varphi(w_2))$: if $\len{w_1}<\len{w_2}$, then
$\len{\varphi(w_1)}<\len{\varphi(w_2)}$ and the order is the same as that of
$w_1$ and $w_2$. Otherwise, if $w_1$ and $w_2$ have the same length, then the
prefixes of $\varphi(w_1)$ and $\varphi(w_2)$ also match, and the lower-order
bits contributed by the individual $\Delta$ cannot change the numeric
ordering, thus leaving the order of the images to be determined by the order
of $w_1,w_2$. This means that $w\leq w_x$ implies
$\varphi(w)\leq\varphi(w_x)$, so that we find
\begin{equation}\label{eqn:lower-bound-1}
    \dens_{SQ}(gn(w))\leq \dens_{L_0}(gn(\varphi(w))).
\end{equation}

We shall use \eqref{eqn:lower-bound-1} to lower-bound $\dens_{L_0}(y)$
asymptotically for sufficiently large $y\in\N$. To this end, let us change
variables in inequality \eqref{eqn:lower-bound-1}, using Figure
\ref{fig:lower-bound-illustration} with $L=SQ$ to guide our intuition:
Equivalently to letting $y$ be arbitrary, we can take an arbitrary word
$w'\in\Sigma^*$ to define $y$ as $y:=gn(w')$. Since $\varphi$ is not
surjective, we cannot hope to find a preimage for every $w'\in\Sigma^*$, so
we distinguish two cases:\\[\baselineskip]%
\noindent\underline{Case 1 ($y=gn(\varphi(w))$ for some $w\in\Sigma^*$)}: The
preimage $w$ of $w'$ under $\varphi$ is unique since the reduction is
injective. By substitution, we get
\begin{equation}\label{eqn:pre-lower-bound}
    \dens_{SQ}(gn(\varphi^{-1}(w')))\leq\dens_{L_0}(gn(w')).
\end{equation}
For such a $w'=\rho(M'_{\overline{SQ}})\$w\$0^*z$, where $w\in SQ$, it is a
simple matter to extract the preimage $w$, located ``somewhere in the
middle'' of $w'$. Precisely, $w$ is located in the left-most $k$-bit block
(among the total of $\lambda$ such blocks), and has a length equal to $\frac
1{\lambda}\len{w'} - c'$, where the constant $c'\geq c$ accounts for the
length of $\rho(M'_{\overline{SQ}})$, the additional 1 bit from the G\"odel numbering, the separator symbols $\$$, and a
possible remainder of zeroes from the $0^\nu$-trailer (containing the
$\Delta$ towards the next square). 

The preimage $w=\varphi^{-1}(w')$ satisfies
\begin{equation}\label{eqn:preimage}
gn(w)\geq
2^{\len{w}}=2^{\frac 1\lambda\ceil{\log(gn(w'))}-c'} \geq
2^{-c''}\sqrt[\lambda]{gn(w')},
\end{equation}
where $c''$ is a constant again.

Substituting this into \eqref{eqn:pre-lower-bound}, we get
$\dens_{SQ}(2^{-c''}\sqrt[\lambda]{gn(w')})\leq\dens_{L_0}(gn(w'))$. Using
that $\dens_{SQ}(x)\in\Theta(\sqrt x)$ (Lemma \ref{lem:density-of-squares}),
we end up finding that for a constant $D'$ (implied by the $\Theta$), another
constant $\beta$ (dependent on $c''$) and sufficiently long $w'$,
\begin{equation}\label{eqn:density-lower-bound}
    D'\sqrt{2^{-c''}\sqrt[\lambda]{gn(w')}}=D\cdot\sqrt[\beta]{(w')_2}\leq\dens_{L_0}(gn(w')),
\end{equation}
where $D>0$ is yet another constant. Observe that the construction requires
$\lambda\geq 3$ and therefore makes $\beta\geq 6$ (we will use this
observation later).\\[\baselineskip]%
\noindent\underline{Case 2 ($y\neq gn(\varphi(w))$ for all $w\in\Sigma^*$)}:
The key insight here is that for lower-bounding the count (the density
function), it suffices to identify \emph{some} $w$ for which
$gn(\varphi(w))<y=gn(w')$, in which case we get a coarser bound $$\dens_{SQ}(gn(w))\leq
\dens_{L_0}(gn(\varphi(w)))<\dens_{L_0}(y)=\dens_{L_0}(gn(w')),$$ where the
second inequality holds since the density function is monotonously
increasing. A simple and reliable way to find $w$ is the following: for the
moment, let us forget about $w'$ not being in the image set of $\varphi$, and
extract a substring from it exactly like in case 1 before (disregarding that
the prefix and suffix may not have the proper form as under $\varphi$). Then,
we shorten the result by deleting one bit from it (say, the least
significant) and call the so-obtained word $w$. Observe that $\varphi(w)$ is
shorter than $w'$, so that $gn(\varphi(w))<gn(w')$ necessarily. But
$\varphi(w)$ has the preimage $w$, so the same arguments as in the previous
case can be used again, starting from \eqref{eqn:preimage} onwards. The only
difference is the constant $c'+1$ instead of $c'$ (as we deleted one more
bit), and the subsequently new constant $D''$ when we re-arrive at
\eqref{eqn:density-lower-bound} (notice that $w'$ re-occurs in that
expression, since we obtained $w$ from $w'$, and only the length of $w'$ but
not its structure played a role in \eqref{eqn:preimage}).

Since the bound \eqref{eqn:density-lower-bound} takes the same form in both
cases, except for the different constants, we can choose the coarser of the
two as a lower limit in all cases.

\begin{rem}\label{rem:hard-language-remark}
Note that (some of) the constants involved here actually and ultimately
depend (through a chain of implications) on the choices of the two functions
$t$ and $T$. These give rise to the language $L_D$ and determine the
``stopwatch'' that we must attach to the simulation of $M_{\overline{SQ}}$
when reducing the language ${SQ}$ to our hard-to-decide language
$L_D\in\dtime(T)\setminus\dtime(t)$. This in turn controls the overhead for
the reduction function $\varphi$ in Lemma \ref{lem:dtime-t-hardness} and the
magnitude of the constants $\lambda,\beta$, etc.
\end{rem}
Together with Lemma \ref{lem:density-upper-bound}, and after substituting
$y=gn(w')=x$, we can strengthen our previous results into stating:
\begin{lem}\label{lem:hard-language}
Let $t,T$ be as in Assumption \ref{asm:time-hierarchy} and assume $t(n)\geq
n^3$. Then, there exists a language $L_0\in\dtime(T)\setminus\dtime(t)$
together with an integer constant $\beta\geq 6$ and a real constant $d>0$,
and some $x_0\in\N$ for which
\[
    d\cdot\sqrt[\beta]{x}\leq\dens_{L_0}(x)\leq \sqrt x\qquad\text{for all~}x\geq x_0.
\]
\end{lem}

\subsection{Threshold Sampling}\label{sec:threshold-sampling}

As the time to evaluate our sought \ac{OWF} is limited to be polynomial, we
cannot construct yes- and no-instances of $w\stackrel ?{\in} L_0$ by directly
testing a randomly chosen word $w$. Instead, we will sample a set of $m$ such
words in a way that probabilistically assures at least one of them to be in
$L_0$ without having to check membership explicitly. That is, we will
randomly draw elements from the family
\begin{equation}\label{eqn:hard-language}
    \LN =\set{W\subset\Sigma^*: gn(W) \subseteq \set{1,\ldots,N} \land W\cap L_0\neq\emptyset},
\end{equation}
where the role and definition of the size $N$ will be discussed in detail
below.

The hardness of this new language is inherited from $L_0$ as the following
simple consideration shows:
\begin{lem}\label{lem:l-star-hardness}
Let $L_0$ be as in Lemma \ref{lem:hard-language}, let $N>0$ and let $\LN$ be
defined by \eqref{eqn:hard-language}. Then $\LN\in\dtime(N\cdot T)\setminus
\dtime(t)$. Here, the input arguments of $t$ and $T$ are the maximal
bitlengths of the words in $W\in\LN$.
\end{lem}
\begin{proof}
Take $w\in\set{0,1}^*$. If $\LN\in\dtime(t)$, then we could take any fixed
$w^*\notin L_0$, and (in polynomial time) cast $w$ into
$W=\set{w,w^*,w^*,\ldots,w^*}$. Obviously, $w\in L_0\iff W\in \LN$, so
$L_0\in\dtime(t)$, which is a contradiction.

Conversely, $W=\set{w_1,\ldots,w_m}\in\LN$ can be decided by checking
$w_i\stackrel ? \in L_0$ for all $i=1,2,\ldots,m\leq N$, which takes a total
of $\leq N\cdot T$ time. So, $\LN\in\dtime(N\cdot T)$.
\end{proof}

Let us keep $N$ fixed for the moment and take $U\subset\Sigma^*$ as a finite
set (urn) with $N$ elements. Then, sampling from $\LN$ amounts to drawing a
subset $W\subseteq U\subset\Sigma^*$, hoping that the resulting set
intersects $L_0$, i.e., $W\cap L_0\neq\emptyset$. To avoid deciding if
$[\exists w\in W: w\in L_0]$, which would take $O(\abs{W}\cdot T(\max_{w\in
W}\len{w}))$ time, we use a probabilistic method from random graph theory.

The predicate $Q_k(W)$ for a $k$-element subset of words $W\subseteq U$ is
defined as ``true'' if $W\cap L_0\neq\emptyset$ (that is, $Q_N$ is yet
another way of defining $\LN$). In the following, let us slightly abuse our
notation and write $Q_k$ to also mean the \emph{event} that $W\cap
L_0\neq\emptyset$ for a randomly chosen $W\subseteq U$ of cardinality $k$.
The likelihood for $Q_k$ to occur under a uniform distribution is, with
$\abs{U}=N$,
\[
\Pr(Q_k)=\abs{\set{W\subseteq U:
\abs{W}=k, W\cap L_0\neq\emptyset}}/\binom N k.
\]
Hereafter, we omit the
subscript and write only $Q$ whenever we refer to the general property (not
specifically for sets of given size).

Lemma \ref{lem:hard-language} tells us that the element count of $L_0$ up to
a number $0<x<N$ is at least $d\cdot\sqrt[\beta]{x}>0$ and $\leq \sqrt x<N$,
when $x$ and $N$ are sufficiently large. This implies that $Q_k$ is actually
a \emph{non-trivial} property of subsets of $U$ (in the sense of describing
neither the empty nor the full set). Moreover, it is a \emph{monotone
increasing} property, since once $Q_k(W)$ holds, then $Q_k(W')$ trivially
holds for every superset $W'\supseteq W$. As it is known that all monotone
properties have a threshold \cite{Bollobas1986}, we now go on looking for one
explicitly by virtue of the following result:

\begin{thm}[{\cite[Thm.4]{Bollobas1986}}]\label{thm:threshold-function}
Let $Q$ be a nontrivial and monotonously increasing property of subsets of a
set $U$, where $\abs{U}=N$.

Let $m^*(N)=\max\set{k: \Pr(Q_k)\leq 1/2}$, and $\vartheta(N)\geq 1$.
\begin{enumerate}
  \item If $m\leq m^*/\vartheta(N)$, then
    \begin{equation}\label{eqn:underneath-threshold}
    \Pr(Q_m)\leq1-2^{-1/\vartheta},
    \end{equation}
  \item and if $m\geq \vartheta(N)\cdot (m^*+1)$, then
      \begin{equation}\label{eqn:excess-threshold}
        \Pr(Q_m)\geq 1-2^{-\vartheta}
      \end{equation}
\end{enumerate}
\end{thm}
The next steps are thus working out $m^*$ explicitly, with the aid of Lemma
\ref{lem:hard-language}. Our first task on this agenda is therefore
estimating $\Pr(Q_k)$, so as to determine the function $m^*$.

Define $p=\dens_{L_0}(N)/N$ as the fraction of elements of $L_0$ among the
entirety of $N$ elements\footnote{The variable $N$ will later be made
dependent on the input length $\ell$, so that $p$ as defined here is actually
$p(\ell)$ as announced in the abstract.} (cf. Lemma
\ref{eqn:density-bounded-by-argument}) in $\set{1,2,\ldots,N}\subset\N$,
whose corresponding words in $U$ are recovered by virtue of $gn^{-1}$. The
total number of $k$-subsets from $N$ elements is $\binom N k$, among which
there are $\binom{(1-p)N}k$ elements that are \emph{not} in $L_0$ (note that
$(1-p)N$ is an integer). Thus, the likelihood to draw a $k$-element subset
that contains at least one element from $L_0$ is given by
\[
    \frac{\binom N k-\binom{(1-p)N}k}{\binom N k}=1-\frac{\binom{(1-p)N}k}{\binom N k}=\Pr(Q_k).
\]
The threshold obviously depends on $p$ (through the predicate/event $Q_k$
that is determined by it), and is by Theorem \ref{thm:threshold-function}
\begin{equation}\label{eqn:threshold-function}
    m^*(N,p) = \max\set{k: \Pr(Q_k)\leq \frac 1 2}=\max\set{k:\frac{\binom{(1-p)N}k}{\binom N k}\geq \frac 1 2}.
\end{equation}
To simplify matters in the following, let us think of the factorial being
evaluated as a $\Gamma$-function (omitted in the following to keep the
formulas slightly simpler), so that all expressions \emph{continuously}
depend on the involved variables (whenever they are well-defined). This
relaxation lets us work with the real value $\kappa\in\R$ (replacing the
integer $k$ for the moment) that satisfies the identity
\begin{equation}\label{eqn:threshold-equation}
\frac{\binom{(1-p)N}\kappa}{\binom N \kappa}=\frac 1 2 = \frac{(N-\kappa)!((1-p)N)!}{N!((1-p)N-\kappa)!}
\end{equation}
instead of having to look for the (discrete) maximal $k\in\N$ so that
$\Pr(Q_k)\leq 1/2$. The sought integer solution to
\eqref{eqn:threshold-function} is then (relying on the continuity) obtained
by rounding $\kappa$ towards an integer.

Since the expressions $((1-p)N)!$ and $N!$ in the nominator and denominator,
respectively, do not depend on $\kappa$, let us expand the remaining quotient
\begin{align}
    &\frac{(N-\kappa)!}{((1-p)N-\kappa)!}=\frac{(N-\kappa)!}{(N-\kappa-pN)!}\nonumber\\
    &\qquad = (N-\kappa-pN+1)(N-\kappa-pN+2)\cdots(N-\kappa-1)(N-\kappa),\label{eqn:pochhammer-product}
\end{align}
which has exactly $pN$ factors (notice that $pN$ is indeed an integer, since
this is just the element count on the condition $w\in L_0$ for $1\leq
gn(w)\leq N$).

Trivial upper and lower bounds on \eqref{eqn:pochhammer-product} are obtained
by using $pN$-th powers of the largest or smallest term in the product. That
is,
\[
    \frac{((1-p)N)!}{N!}((1-p)N-\kappa+1)^{pN}\leq \underbrace{\frac{(N-\kappa)!((1-p)N)!}{N!((1-p)N-\kappa)!}}_{=:r(\kappa)}\leq \frac{((1-p)N)!}{N!}(N-\kappa)^{pN}.
\]
Equation \eqref{eqn:threshold-equation} can be stated more generally as
solving the equation $r(\kappa)=y$ for $\kappa$, given a right-hand side
value $y$. The bounds on $r(\kappa)$ then imply bounds on the solutions of
equation \eqref{eqn:threshold-equation}, which are
\begin{align*}
1+N(1-p)-\left(\frac{y\cdot N!}{((1-p)N)!}\right)^{\frac 1{pN}}\leq r^{-1}(y)\leq N-\left(\frac{y\cdot N!}{((1-p)N)!}\right)^{\frac 1{pN}}.
\end{align*}
By substituting $y=1/2$ into the last expression, we obtain the sought bounds
\[
    \underbrace{\left\lfloor 1+N(1-p)-\left(\frac 1 2\cdot\frac{N!}{((1-p)N)!}\right)^{\frac 1{pN}}\right\rfloor}_{=:\mu_*(N,p)}\leq k\leq \underbrace{\left\lceil N-\left(\frac 1 2\cdot\frac{N!}{((1-p)N)!}\right)^{\frac 1{pN}}\right\rceil}_{=:\mu^*(N,p)}
\]
The threshold $m^*(N,p)$ is defined as the maximal such $k\in\N$, but must
respect the same upper and lower limits, where the rounding operations on the
bounds ($\lfloor\cdot\rfloor$ and $\lceil\cdot\rceil$) preserve the validity
of the limits when $\kappa$ is rounded towards an integer. 
Thus, the bound is now
\begin{equation}\label{eqn:mu-bounds}
  \mu_*(N,p)\leq m^*(N,p)\leq \mu^*(N,p),
\end{equation}
with functions $\mu_*,\mu^*$ induced by the language $L_0$ through the
parameter $p$.

Our next step is using the bounds obtained on the fraction $p$ of elements in
$L_0$ that fall into the discrete interval $[1,N]=U\subset\N$ to refine the
above bounds on the threshold $m^*$. First, we use Lemma
\ref{lem:hard-language} to bound $p$ as
\begin{equation}\label{eqn:p-bounds}
p_* := d\cdot\frac{\sqrt[\beta]{N}}{N}\leq p\leq \frac{\sqrt N}{N} =: p^*
\end{equation}
for sufficiently large $N$. Furthermore, observe that the threshold
$m^*(N,p)$ is monotonously decreasing in $p$, since the more ``good''
elements (those from $L_0$) we have in the set of $N$, the less elements do
we need to draw until we come across a ``good'' one. Thus, for $p_*\leq p\leq
p^*$, we have
\begin{equation}\label{eqn:m-bounds}
  m^*(N,p^*)\leq m^*(N,p)\leq m^*(N,p_*).
\end{equation}

With this, we define the number $m(N)$ of elements that we draw at random
from $\LN$ as
\begin{equation}\label{eqn:threshold-number}
    m = m(N) := \frac 1{\sqrt[\alpha]{N}}\mu_*(N,p^*),
\end{equation}
for a positive constant $\alpha$ that we will determine later.

Note that $\mu_*$ may in some cases take on negative values, but it is
nonetheless an asymptotic nontrivial (i.e., positive and increasing) lower
bound. A quick limit calculation in Mathematica \cite{wolframresearchincMathematicaVersion122023} confirms that $\lim_{N\to\infty}\mu_*(N,1/\sqrt{N})=\infty$, 
but independently, let us expand the product $N!/((1-p)N)!$ occurring in the
definition of $\mu_*(N,p)$. Take $p=p^*=1/\sqrt{N}$ in
\[
\frac{N!}{((1-p^*)N)!}=\prod_{j=0}^{p^*\cdot N-1}(N-j),
\]
and raise both sides to the $\frac 1{p^*N}$-th power, to reveal that each
factor satisfies $1\leq (N-j)^{1/(p^*N)}\leq N^{1/\sqrt{N}}\to 1$. Likewise,
$\sqrt[p^*N]{1/2}\to 1$ for $N\to\infty$, so that
$\mu_*(N,p^*)\in\Theta(\sqrt{N})$, and we get
\begin{equation}\label{eqn:growth-of-m}
  m(N)\in\Theta(N^{1/2-1/\alpha}),
\end{equation}
where $\alpha>2$ induces a growth towards $+\infty$.

Regardless of whether we wish to draw some $W\in\LN$ or $W\notin\LN$, our
sampling algorithm will in any case output a set $W$ of cardinality $m$. The
difference between an output $W\in\LN$ or $W\notin \LN$ is being made on the
number $N$ of elements from which we draw $W$.

The key step towards sampling $W\notin \LN$ is therefore to thin out $U$ by
dropping elements at random, until the cardinality $N=\abs{U}$ is so small
that $m(\abs{U'},p)$ \emph{exceeds} the threshold $m^*$ (that applies to the
now smaller urn $U$). Otherwise, we choose $U$ so large that $m$
\emph{undercuts} the threshold $m^*$ that applies to the full set $U$.


Specifically, we need to suitably thin out $U$ to $U'$, but retaining the elements over the same range in $\N$ 
so that pulling out
the \emph{same} number of $m$ elements either makes
\eqref{eqn:underneath-threshold} or \eqref{eqn:excess-threshold} from Theorem
\ref{thm:threshold-function} apply. In the following, let the smaller set
$U'$ have $n$ entries, and let the larger set $U$ have $N=n^{2\beta}$
entries\footnote{The choice of $2\beta$ is arbitrary and for convenience, to
ease the algebra and to let the expressions nicely simplify.}, where $\beta$
is the constant from Lemma \ref{lem:hard-language}.
\begin{rem}\label{rem:choice}
Observe that the threshold function $m^*(N,p)$ that applies to sampling from
a set $U$ with $N$ elements must always satisfy $m^*(N,p)\leq N=\abs{U}$. By
choosing $\abs{U}\leq n^{2\beta}$, we assure that the threshold $m^*$ (and
hence also the selection count $m$) is polynomial in $n$.
\end{rem}

To sample\ldots
\begin{itemize}
  \item \ldots a no-instance $W\notin\LN$, we use a set
      $\abs{U}=N=n^{2\beta}$ elements. Let us write $p=\abs{U\cap
      L_0}/\abs{U}$ for the likelihood to hit an element from $L_0$ within
      $U$, then we actually undercut the threshold by drawing
      \[
        m = \frac 1{\sqrt[\alpha]{N}}\mu_*(N,p^*)\leq\mu_*(N,p^*)\stackrel{\eqref{eqn:mu-bounds}}{\leq} m^*(N,p^*),
      \]
      elements (note that $N^{-1/\alpha}<1$ for $N>1$). This gives
      \begin{align*}
        \liminf_{N\to\infty}\frac{m^*(N,p)}m\stackrel{\eqref{eqn:m-bounds}}{\geq}\liminf_{N\to\infty}\frac{m^*(N,p^*)}m&\stackrel{\eqref{eqn:mu-bounds}}{\geq}\liminf_{N\to\infty}\frac{\mu_*(N,p^*)}m\\
        &=\liminf_{N\to\infty}\sqrt[\alpha]N=\infty,
      \end{align*}
      so $m(N)$ asymptotically stays under the threshold $m^*$.
  \item \ldots a yes-instance $W\in\LN$, we cut down the cardinality by a
      factor of $s=n^{2\beta-1}$, i.e., we drop elements from $U$ until
      only $\abs{U'}=\abs{U}/s=n^{2\beta}/s = n$ entries remain. Like
      before, let us write $p'=\abs{U'\cap L_0}/\abs{U'}$ for the
      likelihood to draw a member of $L_0$ from $U'\subset U$, and keep in
      mind that the threshold $m^*$ is designed for the smaller urn with
      only $N/s$ entries, from which we nonetheless draw $m$ elements.

Intuitively, observe that the \emph{relative} amount $p'$ of elements from
$L_0$ within $U'$ remains unchanged (in the limit) upon the drop-out
process, provided that the deletion disregards the specific structure of a
word $w$ (which is trivial to implement).

Formally, we have $p=\Pr(w\in L_0|w\in U)$, and $p'=\Pr(w\in L_0|w\in U')$, where the second probability is taken over both random choices, $w$ \emph{and} the subset $U'$.
The latter is
\begin{align*}
  p'&=\frac{\Pr(w\in L_0\land w\in U')}{\Pr(w\in U')}=\frac{\Pr(w\in L_0\land (w\text{ is selected})\land w\in U)}{\Pr((w\text{ is
selected})\land w\in U)}\\
&=\frac{\Pr(w\in L_0\land w\in U)\Pr(w\text{ is selected})}{\Pr(w\in U)\Pr(w\text{ is selected})}=\Pr(w\in L_0|w\in U)=p,
\end{align*}
where the third equality follows from the selection of $w$ into $U'$ being
stochastically independent of the other events. Later, this is achieved by
specifying  Algorithm \ref{alg:selection} (function \textsc{Select}) to
\emph{not} care about how $w$ looks like or relates to the language $L_0$.

So, there is no need to distinguish the parameter $p$ for $U$ and $U'$ and
we can consider
      \begin{align*}
        0&\stackrel{\eqref{eqn:growth-of-m}}{\leq}\limsup_{N\to\infty}\frac{m^*(N/s,p)}m\stackrel{\eqref{eqn:m-bounds}}{\leq}\limsup_{N\to\infty}\frac{m^*(N/s,p_*)}m\\
        &\stackrel{\eqref{eqn:mu-bounds}}{\leq}\limsup_{N\to\infty}\frac{\sqrt[\alpha]N\cdot\mu^*(N/s,p_*)}{\mu_*(N,p^*)}.
      \end{align*}
      
Observe that $m^*$ here depends only on the numbers $N/s$ and $p$, but not explicitly on the (smaller) urn $U'$. The reason is that the choice of $U'$ is part of the predicate $Q_k$ that the threshold $m^*$ uses. In other words, we are using two different predicates for the large urn $U$ of size $N$ and the small urn $U'$ of size $N/s$: for $U$, the predicate $Q_k(W)$ refers to a subset $W$ containing an element of $L_0$. For the small urn (whose threshold is $m^*$), the predicate is about whether the set $W\cap U'$ contains an element of $L_0$, i.e., in case of the small urn, the predicate itself draws a random subset to evaluate. Hence, even though the absolute value $p'=\abs{L_0\cap U'}/\abs{U'}$ would be different from $p$, we still have $\Pr(Q_k)$ for the small urn taken over the random choices of $U'$ (that $Q_k$ internally makes), yielding the same value $p$ for the small urn as derived above.
      
      We substitute $N=n^{2\beta}, s=n^{2\beta-1}$ and the bounds
      \eqref{eqn:p-bounds}, rearrange terms, and cast the factorials into
      $\Gamma$-functions, which turns the last quotient into (dropping the
      $\floor{\cdot}$ and $\ceil{\cdot}$ to ease matters w.l.o.g.),
      \begin{equation}\label{eqn:quotient}
        \frac{\left(n^{2 \beta}\right)^{1/\alpha} \Biggl(1+n-\overbrace{2^{-\frac{n^{2 \beta-3}}{d}} \left(\frac{\Gamma (n+1)}{\Gamma \left(n-d n^{3-2 \beta}+1\right)}\right)^{\frac{n^{2 \beta-3}}{d}}}^{=:A}\Biggl)}{1+\underbrace{n^{\beta} \left(n^{\beta}-1\right)}_{=:B}-\underbrace{2^{-n^{-\beta}} \left(\frac{\Gamma \left(n^{2 \beta}+1\right)}{\Gamma \left(n^{2 \beta}-n^{\beta}+1\right)}\right)^{n^{-\beta}}}_{=:C}}
      \end{equation}
      where $d>0$ and $\beta>0$ are the constant appearing in Lemma
      \ref{lem:hard-language}. Towards showing that
      $\eqref{eqn:quotient}\in O(n^{-\gamma})$ for some constant
      $\gamma>0$, it is useful to consider the nominator and denominator of
      \eqref{eqn:quotient} separately, as well as the terms $A,B$ and $C$
      therein.

      \underline{Nominator of \eqref{eqn:quotient}}: Towards showing that
		term $A$ is bounded, let us first look at the inner quotient 
		\[
		\frac{\Gamma(n+1)}{\Gamma(n-d\cdot n^{3-2\beta}+1)}=\frac{\Gamma(n+1)}{\Gamma(n+\sigma)}
		\]
		with the value $0<\sigma=1-d\cdot n^{3-2\beta}<1$, and apply Gautschi's inequality \cite[Sec. 5.6.4.]{nationalinstituteofstandartsandtechnologynistNISTDigitalLibrary2023}, to obtain the bounds
		\[
		n^{d\cdot n^{3-2\beta}}\leq\frac{\Gamma(n+1)}{\Gamma(n-d\cdot n^{3-2\beta}+1)} \leq (n+1)^{d\cdot n^{3-2\beta}}.
		\]
		Substituting $b:=2\beta-3>0$ to simplify the terms, we can rewrite the bounds as $n^{d\cdot n^{-b}}=\left(n^{\frac 1{n^b}}\right)^d=\left(\sqrt[n^b]{n}\right)^d$, and respectively, $\left(\sqrt[n^b]{n+1}\right)^d$, both of which converge to 1 as $n\to\infty$, and we get
		\begin{equation}\label{eqn:gamma-quotient}
			\lim_{n\to\infty}\frac{\Gamma(n+1)}{\Gamma(n-d\cdot n^{3-2\beta}+1)}=1.
		\end{equation}
		Now, let us return to the denominator of \eqref{eqn:quotient}, and include the term $2^{-1}$ inside the brackets, i.e.,
		\[
		2^{-\frac{n^{2 \beta-3}}{d}} \left(\frac{\Gamma (n+1)}{\Gamma \left(n-d n^{3-2 \beta}+1\right)}\right)^{\frac{n^{2 \beta-3}}{d}} =
		\left(\frac 1 2\cdot \frac{\Gamma (n+1)}{\Gamma \left(n-d n^{3-2 \beta}+1\right)}\right)^{\frac{n^{2 \beta-3}}{d}},
		\]
		and note that by \eqref{eqn:gamma-quotient}, the inner term inside the bracket will converge to $\frac 1 2$, so that ultimately, for sufficiently large $n$, we have the constant upper bound
		\[
		\frac 1 2\cdot \frac{\Gamma (n+1)}{\Gamma \left(n-d n^{3-2 \beta}+1\right)}\leq \frac 2 3
		\]
		Raising both sides of the inequality to the power of $d\cdot n^{2\beta-3}>0$ we get
		\[
		\left(\frac 1 2\cdot \frac{\Gamma (n+1)}{\Gamma \left(n-d n^{3-2 \beta}+1\right)}\right)^{d\cdot n^{2\beta-3}}\leq \left(\frac 2 3\right)^{d\cdot n^{2\beta-3}}\to 0,
		\]
		as $n\to\infty$, since this also pushes $d\cdot n^{2\beta-3}\to\infty$.
		
		Thus, for some constant $E$, we have the nominator of
		\eqref{eqn:quotient} asymptotically bounded as $\leq
		n^{2\beta/\alpha}(n+E)\in O(n^{2\beta/\alpha+1})$.

    \underline{Denominator of \eqref{eqn:quotient}}: It is a quick matter of calculation in \textsc{Mathematica} \cite{wolframresearchincMathematicaVersion122023} to verify that the terms $B$ and $C$ that both depend on $n$, satisfy $\lim_{n\to\infty}\frac{B - C}{n^{\beta/2}}=\infty$, conditional on $\beta>0$ (which holds in our setting). Hence, the denominator of \eqref{eqn:quotient} grows as $\Omega(n^{\beta/2})$.

    Combining the asymptotic bounds on the nominator and denominator, we
    end up asserting
    \[
        \eqref{eqn:quotient}\in O(n^{2\beta/\alpha+1}\cdot n^{-\beta/2})
    \]
    It is easily discovered that $1+2\beta/\alpha - \beta/2<0$, if
    $\beta>2$ (previously, we noted that $\beta\geq 6$) and
    $\alpha>4\beta/(\beta-2)$. Thus, we are free to put $\alpha:=
    4\beta/(\beta-2) + 2\beta$ in \eqref{eqn:threshold-number} (note that
    $\alpha\geq 18$ since $\beta\geq 6$), to achieve
    \begin{equation}\label{eqn:threshold-case-2}
        \limsup_{N\to\infty}\frac{m^*(N/s,p)}{m(N)}\in O(n^{-\gamma}),
    \end{equation}
    where $\gamma = (\beta-2)^2 / (2\beta)\geq \frac 4 3$. Thus, $m$ grows faster than
    the threshold $m^*$ in this case.
\end{itemize}

Now, let us use \eqref{eqn:underneath-threshold} and
\eqref{eqn:excess-threshold} to work out the likelihoods of sampling an
element from $\LN$ or $\overline{\LN}$, which is the set of sample sets that
do (not) contain a word from $L_0$. In the following, let us write $m(N)$ in
omission of the unknown parameter $p$, since this one is replaced by its
upper approximation $p^*$ that depends on $N$ (through \eqref{eqn:p-bounds}).

Let $\vartheta(N)\geq 1$ (according to Theorem \ref{thm:threshold-function}).
\begin{enumerate}
  \item Case 1 of Theorem \ref{thm:threshold-function} applies if $m(N)\leq
      m^*(N,p)/\vartheta(N)$. This is equivalent to $\vartheta(N)\leq
      m^*(N,p)/m(N)$, so that
      \[
        \frac{m^*(N,p)}{m(N)}\stackrel{\eqref{eqn:m-bounds}}{\geq}
        \frac{m^*(N,p^*)}{m(N)}\stackrel{\eqref{eqn:mu-bounds}}{\geq}
        \frac{\mu_*(N,p^*)}{m(N)}=\frac{\mu_*(N,p^*)}{\frac 1{\sqrt[\alpha]{N}}\mu_*(N,p^*)}=\sqrt[\alpha]N,
      \]
      so that we can take $\vartheta(N)=\sqrt[\alpha]N\geq 1$ (as
      required).

      The likelihood to sample an element from $\LN$ thus asymptotically
      satisfies
      \begin{equation}\label{eqn:sampling-from-L-complement}
        \Pr(Q_m)=\Pr(w\in\LN)\leq 1-2^{-1/\vartheta}=1-2^{-1/\sqrt[\alpha]N}\to 0.
      \end{equation}
  \item Case 2 of Theorem \ref{thm:threshold-function} applies if
      $m(N)\geq\vartheta(N)\cdot(m^*(N)+1)$. From
      \eqref{eqn:threshold-case-2}, we have $m^*(n,p)/m(N)\in
      O(n^{-\gamma})$, where $N=n^{2\beta}$. This, and the previously
      established growth of $m(N)\to\infty$ by \eqref{eqn:growth-of-m},
      reveals that when $N$ (and hence also $n$) becomes large enough,
      \[
        \frac{m^*(n,p)+1}{m(N)}=\frac{m^*(n,p)}{m(N)}+\frac 1{m(N)}\leq F\cdot n^{-\gamma}+\frac 1{m(N)},
      \]
      for a constant $F>0$ implied by the $O(n^{-\gamma})$. Thus,
      \[
        (m^*(n,p)+1)\frac 1{F\cdot n^{-\gamma}+\frac 1{m(N)}}\leq m(N),
      \]
      and so we can take
      \[ \vartheta(n)=\frac 1{F\cdot n^{-\gamma}+\frac
        1{m(N)}}=\frac{n^\gamma}{F+\frac{n^\gamma}{m(N)}},
      \]
      after rearranging terms. To analyze the growth of $\vartheta$, we
      substitute the values for $\gamma=(\beta-2)^2/(2\beta)$ and
      $\alpha=4\beta/(\beta-2)+2\beta$ and use \eqref{eqn:growth-of-m} for
      the asymptotic bound $m(N)\geq G\cdot N^{1/2-1/\alpha}=G\cdot
      n^{\beta-2\beta/\alpha}$ for some constant $G>0$. After some
      algebra, we discover 
      \[
        \vartheta(n)\geq \frac{G n^{\beta+\frac{2}{\beta}-1}}{F G n^{\frac{\beta}{2}+1}+1},
      \]
      and the lower bound is quickly verified (in \textsc{Mathematica}) to grow as $\Theta(n^\gamma)$.

      Therefore, by Theorem \ref{thm:threshold-function}, the likelihood to
      sample from $\LN$ asymptotically satisfies
      \begin{equation}\label{eqn:sampling-from-L}
        \Pr(Q_m)=\Pr(w\in\LN)\geq 1-2^{-\vartheta(n)}\geq 1-2^{-\Theta(n^\gamma)}\to 1.
      \end{equation}
\end{enumerate}

At this point, let us briefly resume our sampling method as Algorithm
\ref{alg:threshold-sampling}. The constants $\alpha$ and $\beta$ will
appearing therein depend on the language $L_0$. Its correctness is
established by Lemma \ref{lem:threshold-sampling} as our next intermediate
cleanup.

\begin{algorithm}[h!]
\caption{Threshold Sampling}\label{alg:threshold-sampling}
\begin{algorithmic}[1]
\Require an input bit $b\in\set{0,1}$ and an integer $n\in\N$.%
\Ensure Output of a random finite set $W\subset\Sigma^*$ whose cardinality is
polynomial in $n$, and which either satisfies $W\cap L_0\neq \emptyset$ or
$W\cap L_0=\emptyset$, with high probability, depending on whether $b=1$ or
$b=0$ was supplied.
\Function{Threshold-Sampling}{$b,n$}%
\State $m\gets n^{-2\beta/\alpha}\cdot\mu_*(n^{2\beta},n^{-\beta})$\label{lbl:threshold-calculation}%
\Comment{subst. $N=n^{2\beta}$ in eqs. \eqref{eqn:p-bounds}, \eqref{eqn:threshold-number}}%
\If{$b=1$}\Comment{for $b=1$, exceed the threshold $m^*$}%
    \State choose $U\subset\set{1,2,\ldots,n^{2\beta}}$ with $\abs{U}=n$\label{lbl:threshold-sampling-select1}\Comment{uniform without replacement}%
\Else\Comment{for $b=0$, undercut the threshold $m^*$}%
    \State $U\gets\set{1,2,\ldots,n^{2\beta}}$%
\EndIf%
\State select $W\subseteq U$ with $\abs{W}=m$\label{lbl:threshold-sampling-select2}\Comment{uniform without replacement}%
\State \Return{$W$}
\EndFunction%
\end{algorithmic}
\end{algorithm}

\begin{lem}\label{lem:threshold-sampling}
Algorithm \ref{alg:threshold-sampling} runs in time in $O(n^{2\beta}\log
n\cdot R(n))$, where $R(n)$ is the time required for the random selection in
lines \ref{lbl:threshold-sampling-select1} and
\ref{lbl:threshold-sampling-select2}. It outputs a set $W$ of cardinality
that is polynomial in $n$ (since it is upper bounded by the algorithm's
running time), which satisfies:
\begin{itemize}
  \item $\Pr(W\cap L_0\neq\emptyset|b=1)\geq 1-2^{-\Omega(n^\gamma)}$, and
  \item $\Pr(W\cap L_0=\emptyset|b=0)\geq 2^{-n^{-2\beta/\alpha}}$,
\end{itemize}
where the (positive) constants $\alpha,\beta$ and $\gamma$ depend only on the
language $L_0$.
\end{lem}
\begin{proof}
The events $W\cap L_0=\emptyset$ or $W\cap L_0\neq \emptyset$ correspond to
the previously predicate/event $Q_m$ and its negation. Thus, the asserted
likelihoods follow from \eqref{eqn:sampling-from-L} and
\eqref{eqn:sampling-from-L-complement}, obviously conditional on the input
bit $b$.

The time-complexity of Algorithm \ref{alg:threshold-sampling} is polynomial
in $n$, since we draw no more than $n^{2\beta}$ elements, each of which has
$\leq \lceil\log(n^{2\beta})\rceil$ bits, where $\beta$ is a constant
determined by $L_0$. Moreover, the calculations in line
\ref{lbl:threshold-calculation} are doable in polynomial time less than
$O(n^{2\beta})\supseteq O(n^6)$, since only basic arithmetic over $\R$ is
required (multiplications, divisions and roots).
\end{proof}

\begin{rem} It may be tempting to think of threshold sampling to be
conceptually flawed here, if the experiment is misleadingly interpreted in
the following sense: assume that we would draw a constant number of balls
from two urns, one with few balls in them, the other containing many balls,
but with the fraction of ``good ones'' being the same in both urns. Then, the
likelihood to draw at least one ``good ball'' should intuitively be the same
upon an equal number of trials. However, it must be stressed that the number
of balls in the larger urn grows asymptotically different (and faster) than
the ball count in the smaller urn. Thus, sticking with a fixed number of
trials in both urns, the \emph{absolute} number of balls that we draw from
either urn is indeed identical, but the \emph{fraction} (relative number) of
balls is eventually different in the long run.
\end{rem}

\subsection{Counting the Random Coins in Algorithm \ref{alg:threshold-sampling}}\label{sec:derandomization}
Since Algorithm \ref{alg:threshold-sampling} relies on picking a set of $m$
elements uniformly without replacement from the set
$\set{1,2,\ldots,n^{2\beta}}$ or a subset thereof, we need to know how well a
bunch of $k$ random bits can approximate such a choice, given that $m$ is not
necessarily a power of two. For the time being, let us call
$\omega\in\set{0,1}^*$ an auxiliary lot of random coins that is (implicitly)
available to Algorithm \ref{alg:threshold-sampling}. Our goal is proving
$\len{\omega}\in\poly(n)$ to verify that the selection is doable by a
probabilistic polynomial-time algorithm, to which we can add $\omega$ as
another input.

Specifically, the problem is to choose a random subset (of size $n$ in line
\ref{lbl:threshold-sampling-select1} or size $m$ in line
\ref{lbl:threshold-sampling-select2} of Algorithm
\ref{alg:threshold-sampling}) from a given total of $N$ elements in $U$. In
the following, let us write $m$ for the size of the selected subset.
Furthermore, assume $U$ to be canonically ordered (as a subset of $\N$).

We do the selection by randomly permuting a vector of indicator variables,
defined with $m$ 1's followed by $N-m$ zeroes (i.e., permute the bits of the
word $1^m0^{N-m}$). The selected subset $W\subset U=\set{u_1,\ldots,u_N}$ is
retrieved from the permuted output $(b_1,\ldots,b_N)\in\set{0,1}^N$ by
including $u_i\in W\iff b_i=1$. This procedure is indeed correct for our
purposes, since every $m$-element subset $W\subseteq U$ corresponds to a word
$w'=\set{0,1}^N$, where $w'$ contains exactly $m$ 1-bits at the positions of
elements that were selected into $W$. The representative word $w'$ can thus
be obtained by permuting the word $w=1^m0^{N-m}$, and we count the number of
permutations $\pi$ that yield $w'=\pi(w)$. There are $N!$ permutations in
total. For any fixed permutation $\pi$, swapping the 1's within their fixed
positions leaves $\pi$ unchanged, so the number $N!$ reduces by a factor of
$m!$ for $m$ 1-bits. Likewise, permuting the $(N-m)$ zero-bits only has no
effect, so another $(N-m)!$ cases are divided out. If our choice of $\pi$ is
uniform, the chance to draw any $m$-element subset by this permutation
approach is therefore given by $(m!(N-m)!)/N!=1\slash\binom N m$, which
matches our assumption for the threshold functions in Section
\ref{sec:threshold-sampling}.

Thus, the random selection of an $m$-element subset boils down to a matter of
producing a random permutation of $N=\abs{U}$ elements. We use a Fisher-Yates
shuffle to do this, which requires a method to select an integer $i$
uniformly at random within a prescribed range $i_{\min}\leq i\leq i_{\max}$.

The necessary random integers are obtained by virtue of the auxiliary string
$\omega$. For a single integer, let us take $k$ bits
$b_1,b_2,\ldots,b_k\in\set{0,1}$ from $\omega$, where the exact count will be
specified later. These $k$ bits define a real-valued random quantity $r$ by
setting $r:=(0.b_1b_2b_3\ldots b_k)_2 = \sum_{i=1}^k b_i\cdot 2^{-i}\in
[0,1)$. Note that $r$ actually ranges within the discrete set $R=\set{j\cdot
2^{-k}: j=0,1,2,\ldots,2^k-1}$. To convert $r$ into a random integer in the
desired range $\set{i_{\min},i_{\min}+1,\ldots,i_{\max}}$, we divide the
interval $[0,1)$ into $i_{\max}-i_{\min}+1$ equally spaced intervals of width
$h=1/(i_{\max}-i_{\min}+1)$, and output the index of the sub-interval that
covers $r$ (the process is very similar to the well-known inversion method to
sample from a given discrete probability distribution). This method only
works correctly if $r$ is a continuously distributed random quantity within
$[0,1)$, and is biased when $r$ has a finite mantissa (i.e., is a rational
value). So, our first step will be comparing the ``ideal'' to the ``real''
setting.

If the sampling were ``ideal'', then $r$ would be continuously and uniformly
($c.u.$) distributed over $[0,1)$. With $h$ being the spacing of $[0,1)$, the
method outputs the index $i_0$ with likelihood
\[
\Pr_{c.u.}(i_0)=\int_{i_0\cdot h}^{(i_0+1)\cdot h}1dt=h.
\]

Next, we consider the event of outputting $i_0$ considering that $r$ is
discrete and uniformly ($d.u.$) distributed over $R$, with the probabilities
$\Pr_{d.u.}(r=j\cdot 2^{-k})=2^{-k}$. The output index is $i_0$ if $r\in
[i_0\cdot h,(i_0+1)\cdot h)$. This interval covers all indices $j$ satisfying
$j\cdot 2^{-k}\geq i_0\cdot h$ and $j\cdot 2^{-k}<(i_0+1)\cdot h$, i.e., all
of which lead to the same output $i_0$. Since each possible $r$ occurs with
the same likelihood $2^{-k}$, we get
\[
\Pr_{d.u.}(i_0) = q = \hspace{-4mm}\sum_{j=\ceil{2^k i_0 h}}^{\ceil{2^k(i_0+1)
h}-1}\hspace{-4mm}2^{-k}=2^{-k}\underbrace{\left(\ceil{2^k\cdot(i_0+1)\cdot
h}-\ceil{2^k i_0 h}\right)}_{=:D}
\]
Consider the approximation $\tilde q = 2^{-k}\tilde D$, where $\tilde
D:=2^k\cdot(i_0+1)\cdot h-2^k i_0 h$. Obviously, $|D-\tilde D|\leq 2$, so
that $\abs{q-\tilde q}\leq 2\cdot 2^{-k}=2^{-k+1}$, and therefore, since
$\tilde q = h=\Pr_{c.u.}(i_0)$,
\begin{equation}\label{eqn:single-integer-accuracy}
  \big|\!\Pr_{d.u.}(i_0)-\Pr_{c.u.}(i_0)\big|=\big|\!\Pr_{d.u.}(i_0)-h\big|\leq 2^{-k+1},
\end{equation}
where $i_0$ is an arbitrary integer in the prescribed range $\set{i_{\min},
i_{\min}+1, \ldots, i_{\max}}$, and $k$ is the number of bits in the value
$r=0.b_1b_2\ldots b_k$, which determines the output $i_0$ as $i_0\gets
\floor{r/h}$ for $h=1/(i_{\max}-i_{\min}+1)$.

For the complexity of this procedure, note that all these operations are
doable in polynomial time in $k, \log(i_{\min})$ and $\log(i_{\max})$. Let
us now turn back to the problem of producing a ``almost uniform'' random
permutation by the Fisher-Yates algorithm. In essence, the sought permutation
is created by choosing the first element $\pi(1)$ from the full set of $N$
elements, then retracting $\pi(1)$ from $U$, and choosing the second element
from the remaining $N-1$ elements, and so forth.

If we denote the so-obtained sequence of integers as $i_N, i_{N-1}, \ldots,
i_1$, a uniform choice of the permutation means to draw any possible such
sequence with likelihood
\begin{equation}\label{eqn:uniformly-random-permutation}
  \Pr_{\text{unif}}(i_N,i_{N-1},\ldots,i_1)=\prod_{j=0}^{N-1}\frac 1{N-j},
\end{equation}
since the bits taken from $\omega$ to define $r$ are stochastically
independent in each round.

Our current task is thus comparing this likelihood to the probability of
drawing the same sequence under random choices made upon repeatedly taking
chunks of $k$ bits from the auxiliary input $\omega$. As a reminder of this,
let us replace the measure $\Pr_{d.u.}$ by $\Pr_{\omega}$ in the following,
and keep in mind that the two are the same (based on the procedure described
before).

Note that the output in the $j$-th step is the integer $i_j$ that satisfies
$|\Pr_{\omega}(i_j)-h_j|<2^{-k+1}$, where $h_j$ is the spacing of the
interval (determined by the size of the urn from which we draw; in the $j$-th
step, we have $h_j=1/(N-j)$).

Since the construction of every $i_{j+1}$ is determined by a fresh and
stochastically independent lot of $k$ bits from $\omega$, we have

\begin{align}\label{eqn:joint-likelihood-fully-decomposed}
\Pr_{\omega}(i_N,i_{N-1},\ldots,i_1) = \prod_{j=0}^{N-1}\Pr_{\omega}(i_{j+1}).
\end{align}
Next, we shall pin down the number $k$, which determines how accurate
\eqref{eqn:joint-likelihood-fully-decomposed} approximates
\eqref{eqn:uniformly-random-permutation}. Fix $k=N^2+2$, so that
(asymptotically in $N$ and hence $k$) for $0\leq j<N$,
\[
2^{-k+1}<2^{-N^2}<2^{-N}\frac 1 N<2^{-N}\cdot\frac 1{N-j}.
\]
Combining this with \eqref{eqn:single-integer-accuracy} and recalling that
$h=1/(N-j)$, we can bound every term in
\eqref{eqn:joint-likelihood-fully-decomposed} as
\begin{align*}
    \Pr_{\omega}(i_j)&\in\left(\frac 1{N-j}-2^{-N}\cdot\frac 1{N-j},\frac 1{N-j}+2^{-N}\cdot\frac 1{N-j}\right)\\
&=\left(\frac 1{N-j}\cdot\left[1-2^{-N}\right],\frac 1{N-j}\cdot\left[1+2^{-N}\right]\right).
\end{align*}
In particular, this gives a nontrivial lower bound\footnote{indeed, also an
upper bound, but this is not needed here.} to
\eqref{eqn:joint-likelihood-fully-decomposed},
\begin{align}
\Pr_{\omega}(i_N,i_{N-1},\ldots,i_1)&\geq \big(1-2^{-N}\big)^N\cdot\prod_{j=0}^{N-1}\frac 1{N-j}\nonumber\\
&=\big(1-2^{-N}\big)^N\cdot\Pr_{\text{unif}}(i_N,i_{N-1},\ldots,i_1)\label{eqn:relative-error-lower-bound}.
\end{align}

The important part herein was the setting of $k=N^2+2$ to draw a single
integer. Our goal was the selection of a set of $m(N)\leq N$ out of $N$
elements, and we need $N$ integers to get the entire permutation of
$\set{1,2,\ldots N}$. So, the total lot of necessary i.i.d. random coins in
$\omega$ is $N\cdot k\leq N\cdot (N^2+2)\in O(N^3)$. Since, $N\leq
n^{2\beta}$ in every case (see Algorithm \ref{alg:threshold-sampling}), we
have $\len{\omega}\leq\poly(n)$ as claimed.

For another intermediate cleanup, let us compile our findings into the
probabilistic selection Algorithm \ref{alg:selection} (that is actually a
deterministic procedure with an auxiliary lot $\omega$ of random coins). Note
that our specification of the algorithm returns the (potentially empty)
remainder of unused bits in $\omega$. This will turn out necessary over
several invocations of the selection algorithm during the threshold sampling,
to avoid re-using randomness there.

\begin{algorithm}[h!]
\caption{Uniformly Random Selection}\label{alg:selection}
\begin{algorithmic}[1]
\Require a set $U=\set{u_1,u_2,\ldots,u_N}$ of cardinality $N$, and a string
$\omega$ consisting of $\geq N^3+2N$ i.i.d. uniform random bits.

\Ensure a uniformly random subset $W\subseteq U$ of cardinality $m$ for which
\eqref{eqn:relative-error-lower-bound} holds, and the rest of the random bits
in $\omega$ (that have not been used).

\Function{Select}{$m,U,\omega$}%
\State $W\gets \emptyset; N\gets \abs{U}$%
\State $k \gets N^2+2$%
\For{$j=0,1,\ldots, N-1$}\Comment{construct the permutation}
    \State $r\gets (0.b_1b_2\ldots b_k)_2$\label{lbl:setting-r}\Comment{$\omega=b_1b_2\ldots
    b_kb_{k+1}b_{k+2}\ldots$}
    \State $\omega\gets b_{k+1}b_{k+2}\ldots$\Comment{delete used bits from
    $\omega$}
    \State $h\gets 1/(N-j)$%
    \State define $\pi(j) := \floor{r/h}$%
\EndFor%
\State $w' \gets \pi(1^m0^{N-m})$\label{lbl:permute}\Comment{$w' = b_1b_2\ldots b_N$}%
\State \textbf{for} $i=1,2,\ldots,N$, put $u_i\in W\iff b_i=1$\label{lbl:selection-output}%
\State \Return{$(W,\omega)$}%
\EndFunction%
\end{algorithmic}
\end{algorithm}

To finally specify Algorithm \ref{alg:threshold-sampling} with the auxiliary
input $\omega$, we simply need to replace the truly random and uniform
selection of subsets in Algorithm \ref{alg:threshold-sampling} (lines
\ref{lbl:threshold-sampling-select1} and
\ref{lbl:threshold-sampling-select2}) by our described selection procedure
based on random coins from $\omega$, which is algorithm \textsc{Select}. For
convenience of the reader, the result is given as Algorithm
\ref{alg:derandomized-threshold-sampling}.

\begin{algorithm}[h!]
\caption{Probabilistic Threshold
Sampling}\label{alg:derandomized-threshold-sampling}
\begin{algorithmic}[1]
\Require a bit $b\in\set{0,1}$, an integer $n\in\N$, and a word
$\omega\in\set{0,1}^{\poly(n)}$.%
\Ensure a random finite set $W\subset\Sigma^*$ whose cardinality is
polynomial in $n$, and rest of the bits in $\omega$ that have not been used.%
\Function{PTSamp}{$b,n,\omega$}%
\State $m\gets n^{-2\beta/\alpha}\cdot\mu_*(n^{2\beta},n^{-\beta})$%
\Comment{subst. $N=n^{2\beta}$ in eqs. \eqref{eqn:p-bounds}, \eqref{eqn:threshold-number}}%
\If{$b=1$}\Comment{exceed the threshold $m^*$}%
    \State
    $(U,\omega)\gets\textsc{Select}(n,\set{1,2,\ldots,n^{2\beta}},\omega)$\Comment{restrict $U$}%
\Else\Comment{for $b=0$, undercut the threshold $m^*$}%
    \State $U\gets\set{1,2,\ldots,n^{2\beta}}$\Comment{use all of $U$}%
\EndIf%
\State $(W,\omega) \gets \textsc{Select}(m,U,\omega)$\Comment{choose $m$ elements}%
\State \Return{$(W,\omega)$}
\EndFunction%
\end{algorithmic}
\end{algorithm}

To lift Lemma \ref{lem:threshold-sampling} to the new setting of Algorithm
\ref{alg:derandomized-threshold-sampling}, let us apply
\eqref{eqn:relative-error-lower-bound} to the likelihood of the events $Q_m$
and $\neg Q_m$, which mean ``hitting an element from $\LN$ within a selection
of $m$ elements'', or not, respectively.

Specifically, we are interested in the likelihoods $\Pr_{\omega}(Q_m)$ and
$\Pr_{\omega}(\neg Q_m)$, which under ``idealized'' sampling are bounded from
below by \eqref{eqn:sampling-from-L-complement} and
\eqref{eqn:sampling-from-L}, but are now to be computed under the sampling
using the auxiliary string $\omega$.

For the general event $Q\in\set{Q_m,\neg Q_m}$, let us write the likelihood
$\Pr_{\omega}(Q)$ as a sum over all its (disjoint) atoms, we get
\begin{align*}
\Pr_{\omega}(Q)&=\sum_{A\in Q}\Pr_{\omega}(A)\stackrel{\eqref{eqn:relative-error-lower-bound}}{\geq} \sum_{A\in Q}\big(1-2^{-N}\big)^N\Pr_{\text{unif}}(A)\\
&=\big(1-2^{-N}\big)^N\sum_{A\in Q}\Pr_{\text{unif}}(A)=\big(1-2^{-N}\big)^N\cdot\Pr_{\text{unif}}(Q).
\end{align*}
So, we can re-state Lemma \ref{lem:threshold-sampling} in its new version,
using \eqref{eqn:relative-error-lower-bound}. The (yet unknown) term $R(n)$
measuring the running time for a selection is obtained by inspecting
Algorithm \ref{alg:selection}: with $N=\abs{U}$, we need $O(N)$ iterations to
construct the permutation, needing $O(N^2)$ bits per iteration of the loop
(line \ref{lbl:setting-r}), and another $O(N)$ iterations to permute and
deliver the output (lines \ref{lbl:permute} and \ref{lbl:selection-output}).
The overall running time thus comes to $O(N^3)$. Since the selection in
Algorithm \ref{alg:threshold-sampling} is done on sets of size
$N=n^{2\beta}$, the effort for a selection is $R(n)\in O(n^{6\beta})$ in
Lemma \ref{lem:threshold-sampling}.
\begin{lem}\label{lem:Probabilistic-Threshold-sampling}
Algorithm \ref{alg:derandomized-threshold-sampling} runs in polynomial time
$O(n^{8\beta}\log n)$ and outputs a set $W$ of cardinality polynomial in $n$
(since it is upper bounded by the algorithm's running time), which satisfies:
\begin{align}
  \Pr(W\cap L_0\neq\emptyset|b=1)&\geq
      (1-2^{-{n^{2\beta}}})^{n^{2\beta}}\cdot(1-2^{-\Omega(n^\gamma)})\label{eqn:det-sampling-error-1}\\
  \Pr(W\cap L_0=\emptyset|b=0)&\geq
      (1-2^{-{n^{2\beta}}})^{n^{2\beta}}\cdot 2^{-n^{-2\beta/\alpha}}\label{eqn:det-sampling-error-2},
\end{align}
where the (positive) constants $\alpha,\beta$ and $\gamma$ depend only on the
language $L_0$.
\end{lem}

\begin{rem}
It is of central importance to note that our proof is based on random draws
of sets that provably contain the sought element with a probabilistic
assurance but without an explicit certificate. In other words, although the
sampling guarantees high chances of the right elements being selected and
despite that we know what we are looking for, we cannot efficiently single
out any particular output elements, which was the hit.
\end{rem}

\subsection{Partial Bijectivity}\label{sec:bijectivity}
By Definition \ref{def:owf}, we can consider the input string $w=b_1b_2\ldots
b_\ell \in\set{0,1}^\ell$ to our (to be defined) \ac{OWF} as a bunch of
i.i.d. uniformly random bits, which we can split into a prefix word
$v=b_1\ldots b_n$ of length $\len{v}=n$ and a postfix $\omega=b_{n+1}\ldots
b_\ell$ so that $\len{\omega}\geq N^3+2N$ (as Algorithm \ref{alg:selection}
requires), with $N=n^{2\beta}$. For sufficiently large $\ell$, this division
yields nonempty strings $v$ and $\omega$, when $n$ is set to $n(\ell):=
\max\set{i\in\N: i^{6\beta} + 2i^{2\beta} + i \leq \ell}$, i.e., the largest
length $n$ for which the remainder $\omega$ is sufficient to do the
probabilistic sampling under Algorithm
\ref{alg:derandomized-threshold-sampling}. It is easy to see that
$n(\ell)\to\infty$ as $\ell\to\infty$, and the time-complexity to compute
$n(\ell)$ is $\poly(\ell)$.

Based on Figure \ref{fig:owf-construction}, our \ac{OWF} $f_\ell$ will then
be defined on $w$ as a bitwise mapping of the prefix $v$ under the
probabilistic threshold sampling Algorithm
\ref{alg:derandomized-threshold-sampling}, which ``encodes'' the $1/0$-bits
of $v$ as yes/no-instances of the decision problem $\LN$. Formally, this is:

\begin{equation}\label{eqn:owf-preprocessing}
\left.
\begin{array}{l}
 \text{for } i=1,2,\ldots,n,\\
 \qquad (W_i,\omega) \gets \textsc{PTSamp}(b_i, n, \omega)\\
f_\ell(w) = f_\ell(b_1\ldots b_nb_{n+1}\ldots b_\ell) := (W_1,\ldots,W_n).
\end{array}\right\}
\end{equation}

Our objective in the following is \emph{partial bijectivity} of that mapping,
in the sense of assuring that the first bit of the unknown input prefix $w$
to $f_\ell$ can uniquely be computed from the image $f_\ell(w)$, even though
$f_\ell$ may not be bijective. This invertibility will of course depend on
the parameter $\ell$, which determines the value $n$ and through it controls
the likelihood for a sampling error (as quantified by Lemma
\ref{lem:Probabilistic-Threshold-sampling}). If this likelihood is
``sufficiently small'' in the sense that the next Lemma
\ref{lem:bijectivity-criterion} makes rigorous, then $f_\ell$ is indeed
invertible on its first input bit.

\begin{lem}\label{lem:bijectivity-criterion}
Let $X,Y$ be finite sets of equal cardinality and let $f:X\to Y$ be a
deterministic function, where $\Pr(f(x)=f(x'))\leq p$ for any distinct
$x,x'\in X$ drawn uniformly at random. If $p<\frac 2{\abs{X}^2-\abs{X}}$,
then $f$ is bijective.
\end{lem}
\begin{proof}
It suffices to show injectivity of $f$, since the finiteness of $X$ and $Y$
together with $\abs{X}=\abs{Y}$ and injectivity of $f$ implies surjectivity
and hence invertibility of $f$. Towards the contradiction, assume that two
values $x\neq x'$ exist that map onto $z=f(x)=f(y)$, i.e., $f$ is not
injective. Call $p$ the probability for this to happen, taken over all pairs
$(x,x')\in X\times X$ (the probability can be taken as relative frequency;
the counting works since $f$ is deterministic). This means that
$p=\Pr(f(x)=f(x'))\geq 1/ \binom{\abs{X}}2$, which contradicts our
hypothesis.
\end{proof}

Towards applying Lemma \ref{lem:bijectivity-criterion}, we will focus on the
first coordinate function
\[
f_{\ell,1}(b_1) =
\textsc{PTSamp}(b_1,n,\omega),
\]
with inputs as specified above (see \eqref{eqn:owf-preprocessing}).

Since the input to $f_{\ell,1}$ is a pair
$(b_1,\omega)\in\set{0,1}\times\set{0,1}^*=X$, we can partition the pre-image
space $X$, based on the first input bit, into the two-element family $\X_\ell
= \set{[0],[1]}$ with $[b]:=\set{b\omega: \omega\in\set{0,1}^{\ell-1}}$ for
$b=0,1$. In this view, we can think of $f_{\ell,1}$ acting
\emph{deterministically} on $\X$, since the randomness $\omega$ used in
Algorithm \ref{alg:derandomized-threshold-sampling} is supplied with the
input, but the equivalence class is the same for all possible $\omega$. For
the sake of having $f_{\ell,1}$ map into a two-element image set, we will
partition the output set $\mathcal{Z}_m=\set{W=\set{w_1,\ldots,w_{m}}:
w_i\in\Sigma^*\,\forall i}=f_{\ell,1}(\Sigma^\ell)$ with $m=m({N})$ in a
similar manner as $\Y_{\ell}=\set{\Y_{\ell}^{(0)}, \Y_{\ell}^{(1)}}$ with
$\Y_{\ell}^{(0)}:=f_{\ell,1}([0])$ and $\Y_{\ell}^{(1)}=f_{\ell,1}([1])$ for
$[0],[1]\in\X_\ell$. Then, $f_{\ell,1}:\X_\ell\to\Y_\ell$, with
$\abs{\X_\ell}=\abs{\Y_\ell}=2$ for every $\ell$.

Take $x=0\omega\neq x'=1\omega'$ as random representatives of $[0]$ and
$[1]$. The likelihood of the coincidence $f_{\ell,1}(x)=f_{\ell,1}(x')$ is
then determined by the random coins $\omega,\omega'$ in $x$ and $x'$, which
directly go into Algorithm \ref{alg:derandomized-threshold-sampling}. The
partition induces an equivalence relation $\sim$ on the image set of
$f_{\ell,1}$, by an appeal to which we can formulate the criterion of Lemma
\ref{lem:bijectivity-criterion},
\begin{align}
\Pr_{\omega,\omega'}&(f_{\ell,1}([0])\sim f_{\ell,1}([1]))=\Pr_{\omega,\omega'}\left(\substack{\phantom{\lor}\left[f_{\ell,1}([0])\in\Y_\ell^{(0)}\land f_{\ell,1}([1])\in\Y_\ell^{(0)}\right]\\\lor\left[f_{\ell,1}([0])\in\Y_\ell^{(1)}\land f_{\ell,1}([1])\in\Y_\ell^{(1)}\right]}\right)\nonumber\\
&\leq\Pr_{\omega,\omega'}\left(f_{\ell,1}([0])\in\Y_\ell^{(0)}\land f_{\ell,1}([1])\in\Y_\ell^{(0)}\right)\nonumber\\
&\qquad+ \Pr_{\omega,\omega'}\left(f_{\ell,1}([0])\in\Y_\ell^{(1)}\land f_{\ell,1}([1])\in\Y_\ell^{(1)}\right)\nonumber\\
&\leq\Pr_{\omega}\left(f_{\ell,1}([0])\in\Y_\ell^{(1)}\right)+\Pr_{\omega'}\left(f_{\ell,1}([1])\in\Y_\ell^{(0)}\right)\label{eqn:invertibility-criterion},
\end{align}
where the first inequality is the union bound, and the second inequality
follows from the general fact that for any two events $A,B$, we have
$\Pr(A\land B)\leq \min\set{\Pr(A),\Pr(B)}$.

The last two probabilities have been obtained along the proof of Lemma
\ref{lem:Probabilistic-Threshold-sampling}, since:
\begin{enumerate}
  \item $f_{\ell,1}([0])$ means sampling towards avoidance of drawing an
      element from $L_0$, the likelihood of which is bounded by
      \eqref{eqn:det-sampling-error-2}. Therefore,
      $\Pr(f_{\ell,1}([0])\in\Y_\ell^{(1)})=\Pr(W\cap
      L_0\neq\emptyset|b=0)\leq 1-(1-2^{-{n^{2\beta}}})^{n^{2\beta}}\cdot
      2^{-n^{-2\beta/\alpha}}$
  \item $f_{\ell,1}([1])$ means sampling towards drawing at least one
      element from $\LN$, which by \eqref{eqn:det-sampling-error-1},
      implies $\Pr(f_{\ell,1}([1])\in\Y_\ell^{(0)})=\Pr(W\cap
      L_0=\emptyset|b=1)\leq
      1-(1-2^{-{n^{2\beta}}})^{n^{2\beta}}\cdot(1-2^{-\Omega(n^\gamma)})$.
\end{enumerate}
Substituting these bounds into \eqref{eqn:invertibility-criterion}, the
hypothesis of Lemma \ref{lem:bijectivity-criterion} is verified if we let
$\ell$ grow so large that the implied value of $n$ satisfies
\[
2-\underbrace{(1-2^{-{{n}^{2\beta}}})^{{n}^{2\beta}}}_{\to 1}
\cdot\Big[\underbrace{(1-2^{-\Omega(n^\gamma)})}_{\to 1} +
\underbrace{2^{-{n}^{-2\beta/\alpha}}}_{\to 1}\Big]<\frac 2{\abs{\X}^2-\abs{\X}}=1,
\]
to certify the invertibility of $f_{\ell,1}$.

Note that Lemma \ref{lem:bijectivity-criterion} asserts only that the first
bit of the preimage $w$ is determined by the image under $f_{\ell,1}$, but
does so nonconstructively. That is, we only know the the action of
$f_{\ell,1}$ to be either
\begin{align}
  \phantom{or\quad}[0]\mapsto \Y_\ell^{(0)}, &\quad [1]\mapsto \Y_\ell^{(1)}\label{eqn:mapping1}\\
\text{or}\quad[0]\mapsto \Y_\ell^{(1)}, &\quad [1]\mapsto \Y_\ell^{(0)}\label{eqn:mapping2},
\end{align}
where even the possibility of $f_{\ell,1}^{-1}$ being defined alternatingly
by both, \eqref{eqn:mapping1} and \eqref{eqn:mapping2}, is not precluded.

Conditional on \eqref{eqn:mapping1}, the inverse $f_{\ell,1}^{-1}$ is
actually the characteristic function $\chi_{\LN}$ of the language $\LN$ (as
defined in Lemma \ref{lem:l-star-hardness}). However, claiming that
$f_{\ell,1}^{-1}=\chi_{\LN}$ uniformly holds is only admissible if
\eqref{eqn:mapping1} holds for the inputs to $f_{\ell,1}$. We define this to
be an event on its own in the following, denoted as
\begin{equation}\label{eqn:correct-mapping}
  E_\ell := \set{w=(b_1,\ldots,b_\ell)\in\set{0,1}^\ell:
f_{\ell,1}(w)\in\Y_\ell^{(b_1)}}.
\end{equation}
By construction, the conditioning on $E_\ell$ is not too restrictive and even
fading away asymptotically, as told by the next result:
\begin{lem}\label{lem:conditioning}
Let the event $E_\ell$ be defined by \eqref{eqn:correct-mapping}, and let $A$
be any event in the same probability space as $E_\ell$. Then,
$\lim_{\ell\to\infty}\Pr(A|E_\ell)=\Pr(A)$.
\end{lem}
\begin{proof}
Observe that $\Pr(\neg E_\ell)=\Pr(f_{\ell,1}([0])\in\Y_\ell^{(1)}\lor
f_{\ell,1}([1])\in\Y_\ell^{(0)})$, and that the last expression, as was shown
before, tends to zero as $\ell\to\infty$. Then, expanding $\Pr(A)$
conditional on $E_\ell$ and $\neg E_\ell$ into
$\Pr(A)=\Pr(A|E_\ell)\Pr(E_\ell)+\Pr(A|\neg E_\ell)\Pr(\neg E_\ell)$, the
claim follows from $\Pr(\neg E_\ell)\to 0$ and $\Pr(E_\ell)=1-\Pr(\neg
E_\ell)\to 1$ when $\ell\to\infty$.
\end{proof}

Conditional on $E_\ell$, we can state that a circuit computing
$f_{\ell,1}^{-1}$ equivalently decides $\LN$. But Lemma
\ref{lem:l-star-hardness} asserts this decision to be impossible with less
than a certain minimum of $t$ steps. This, together with Lemma
\ref{lem:conditioning}, will be the fundament for the concluding arguments in
the next section.

\subsection{Conclusion on the Existence of Weak \acp{OWF}}\label{sec:weak-owf-exist}
Closing in for the kill, let us now return to the original problem of proving
non-emptiness of Definition \ref{def:owf}.

In the following, we let $\ell\in\N$ be arbitrary. Our final \ac{OWF}
$f_{\ell}$ will be a slightly modified version of
\eqref{eqn:owf-preprocessing},
\begin{equation}\label{eqn:final-owf}
\left.
\begin{array}{rcl}
f_\ell: \set{0,1}^\ell&\to& \Y_\ell^n,\\
(b_1,\ldots,b_n,b_{n+1},\ldots,b_\ell)&\mapsto& ([W_1],[W_2],\ldots,[W_n]),\\
\text{where }n &:=& \max\set{i\in\N: i^{6\beta}+2i^{2\beta}+i\leq\ell},\\
\omega_0 &:=& b_{n+1}b_{n+2}\ldots b_\ell\in\set{0,1}^{\ell-n},~\text{and}\\
(W_i,\omega_i) &:=& \textsc{PTSamp}(b_i,\omega_{i-1})\text{ for }i=1,2,\ldots,n.
\end{array}
\right\}
\end{equation}

We proceed by checking the hypothesis of Definition \ref{def:weak-owf}
one-by-one to verify that \eqref{eqn:final-owf} really defines a weak
\ac{OWF}:

\begin{itemize}
  \item Polynomially related input and output lengths: let the length of
      the output be $n'$, and note that $n'\leq\len{w}\cdot n^{2\beta}$ in
      every case. Assume that all words in the set $U$, from which
      Algorithm \ref{alg:derandomized-threshold-sampling} samples, are
      padded up to the maximal bitlength needed for (the numeral)
      $n^{2\beta}$. Since $n\leq \ell$, we get $n'=n\cdot n^{2\beta}\leq
      \ell^{2\beta+1}$. Thus, $n'\leq\poly(\ell)$. Conversely, we can solve
      for $\ell$ to get $\ell\leq (n')^{1/(2\beta+1)}$, and
      $\ell\leq\poly(n')$. Thus, $f_\ell$ has polynomially related input
      and output length.
  \item Length regularity of $f_\ell$: Evaluating $f_\ell(w)$ means
      sampling from a domain $U$ whose maximal element has magnitude $\leq
      n^{2\beta}$, where $n$ satisfies the bound $\ell\geq
      n^{6\beta}+2n^{2\beta}+n$. Since the numeric range of $U$ is
      determined by the length of the input, equally long inputs result in
      equally long outputs of $f_\ell$. Thus, $f_\ell$ is length regular.
  \item $f_\ell$ can be computed by a deterministic algorithm in polynomial
      time: note that $f_\ell$ is defined by algorithm
      \ref{alg:derandomized-threshold-sampling}, which is actually a
      deterministic procedure that takes its random coins from its input
      only. Furthermore, it runs in polynomial time in $n$ (by lemma
      \ref{lem:Probabilistic-Threshold-sampling} and the fact that $n$ in
      \eqref{eqn:final-owf} can be computed in time $\poly(\ell)$). Since
      $n\leq \ell$, the overall time-complexity is also polynomial in
      $\ell$, so Definition \ref{def:owf} is satisfied up to including
      condition 1, since the (component-wise) equality of $f_\ell(w)$ and
      the output of Algorithm \ref{alg:derandomized-threshold-sampling}
      demanded by Definition \ref{def:owf} is here in terms of equivalence
      classes and not their (random) representatives.
\end{itemize}

It remains to verify condition 2 of Definition \ref{def:owf}, and Definition
\ref{def:weak-owf}, respectively. This amounts to exhibiting a polynomial $q$
so that for any polynomial\footnote{To avoid confusion with the relative
density $p$ that was used in Section \ref{sec:threshold-sampling}, we refrain
from denoting the polynomial $p$ appearing in Definition \ref{def:weak-owf}
explicitly, and write $\poly(\ell)$ here instead (also to remind that the
choice of $p$ would be arbitrary anyway).} $\poly(\ell)$ (determining the
size of the inversion circuit $C$), our constructed function is
$(1-1/q(\ell),\poly(\ell))$-one-way for sufficiently large $\ell$. Observe
the order of quantifiers in Definition \ref{def:weak-owf}, which allows the
minimal magnitude of $\ell$ to depend on all the parameters $(\eps,S)$ of the
definition, especially the polynomials $q$ and $\poly(\ell)$ that define
$\eps=1-1/q$ and $S=\poly(\ell)$. We will keep this in mind in the following.
Throughout the rest of this work, let $\C$ denote the class of all circuits
of size polynomial in $\ell$.

Note that even though $f_\ell$ is not (required to be) bijective, the first
bit $b_1$ in the unknown preimage $w=b_1b_2\ldots
b_{\ell}\in\set{0,1}^{\ell}$ is nevertheless uniquely pinned down upon
knowledge of the first set-valued entry in our \ac{OWF}'s output
$\set{\set{w_1,\ldots,w_{m}},\ldots}$ (where $m$ is computed internally by
Algorithm \ref{alg:derandomized-threshold-sampling}). So, to clear up things
and prove $f_\ell$ to be one-way, let us become specific on the language
$L_0$ that we will use. To define this hard-to-decide language, we
instantiate $t,T$ as follows, where our choice is easily verified to satisfy
Assumption \ref{asm:time-hierarchy}:
\begin{itemize}
  \item Let $L_x[a,b]$ be the well-known subexponential yet superpolynomial
      functional $L_x[a,b]:=2^{a\log(x)^b(\log\log x)^{1-b}}$, and put
  \begin{equation}\label{eqn:superpolynomial-t}
   t(x):=L_x[1, 1/2].
  \end{equation}
  \item $T(x) := 2^x$, which is time-constructible.
\end{itemize}
Furthermore, let $C\in\C$ be an arbitrary circuit of polynomial size
$S(\ell)$, which ought to compute any preimage in $f_\ell^{-1}(f_\ell(w))$,
given $f_\ell(w)$ for $w\in\set{0,1}^\ell$ chosen uniformly at random.

\begin{rem}
Note that constructing the diagonal language $L_D$ with our chosen
superpolynomial function $t$ already prevents any polynomial time machine $M$
from correctly computing a preimage bit. However, we need to be more specific
on the \emph{probability} for such a failure (the construction in the time
hierarchy theorem shows only the necessity of such errors, but not its
frequency).
\end{rem}

The event $[C(f_\ell(w))\in f_\ell^{-1}(f_\ell(w))]$ implies that $C$ must in
particular compute $b_1$ correctly, since $f_\ell$ is bijective on its first
input bit. Conversely, this means that an incorrect such computation implies
the event $[C(f(w))\notin f_\ell^{-1}(f_\ell(w))]$, and in turn
\begin{align}
    \Pr_{w\in\Sigma^\ell}&[C(f_\ell(w))\notin f_\ell^{-1}(f_\ell(w))]\nonumber\\
        &\geq \Pr_{w\in\Sigma^\ell}[C\text{ incorrectly computes }b_1\text{ from }f_\ell(w)]\label{eqn:c-error-bound},
\end{align}
where $b_1$ denotes the first bit in $w$. So, we may focus our attention on
the right hand side probability in the following.

Remember that we constructed our sampling algorithm to output a set
$W_1\in\LN\iff b_1=1$ and $W_1\notin\LN\iff b_1=0$. Despite this, note that a
correct computation of $b_1$ is indeed \emph{not equivalent} to the
computation of the characteristic function $\chi_{\LN}$ of $\LN$, since an
incorrect mapping of $b_1$ on the output equivalence class $f_{\ell,1}=[W_1]$
is nevertheless possible (the sampling made by Algorithm
\ref{alg:derandomized-threshold-sampling} is still probabilistic).

So, to properly formalize the event ``$C$ correctly computes $b_1$'', we must
make our following arguments conditional on the event $E_\ell$ of a correct
mapping, so that
\[
\text{``$C$ correctly computes $b_1$''}\iff C(f_{\ell,1}(w))=\chi_{\LN}(f_{\ell,1}(w))
\]
and in turn
\[
\text{``$C$ incorrectly computes $b_1$''}\iff C(f_{\ell,1}(w))\neq\chi_{\LN}(f_{\ell,1}(w))
\]
are both valid assertions in light of $E_\ell$. Let us consider the second
last likelihood
\[
 \Pr_{w\in\Sigma^\ell}(C=\chi_{\LN}|E_\ell)=\Pr_{w\in E_\ell}(C=\chi_{\LN})
\]
more closely (where the equality is due to the inclusion
$E_\ell\subset\Sigma^\ell$).

If there were a circuit $C\in\C$ that decides $\LN$, then Lemma
\ref{lem:l-star-hardness} (more specifically its proof) gives us an injective
reduction $\psi:L_0\to\LN, w\mapsto (w,w^*,w^*,\ldots)$, where $w^*$ is a
fixed word. Note that $\psi$ can be computed by a polynomial size circuit
(simply by adding hardwired multiple outputs of $w^*$). By this reduction, we
have $w\in L_0\iff \psi(w)\in\LN$, or equivalently,
$\chi_{\LN}(\psi(w))=\chi_{L_0}(w)$. Let $\psi(w)$ be a ``positive case''
(i.e., a word for which $C(\psi(w))=\chi_{\LN}(\psi(w))$ holds), then this
decision is also correctly made for $L_0$, using another polynomial size
circuit $C'=C\circ \psi$. This means that $\Pr_{w\in
E_\ell}(C(\psi(w))=\chi_{\LN}(\psi(w)))\leq\Pr_{w\in
E_\ell}(C'(w)=\chi_{L_0}(w))$, because $\psi$ is injective (otherwise, it
could happen that some instances of $w\stackrel ?\in L_0$ are mapped onto the
same image $\psi(w)$, which could reduce the total count). This leads to the
implication
\begin{align}
  [\exists C\in\C: &\Pr_{w\in E_\ell}(C\text{ decides }\LN)>\eps]\nonumber\\
  &\quad \to[\exists C'\in\C: \Pr_{w\in E_{\ell}}(C'\text{ decides }L_0)>\eps]\label{eqn:l_n-implies-l_0},
\end{align}
where the abbreviation ``$C$ decides $L$'' is a shorthand for $C$ computing
the characteristic function of $L$ (the free variable $\eps>0$ is
$\forall$-quantified, but omitted here to ease our notation).

Similarly, assuming the existence of a circuit $C'\in\C$ that decides $L_0$,
Lemma \ref{lem:dtime-t-hardness} gives us another mapping $\varphi:
\Sigma^*\to \Sigma^*$ for which $\varphi$ modifies the right half of its
input string accordingly so that $\varphi(w)$ becomes a square, while
retaining the left part of $w$ that determines the membership of $w$ in
$L_D$. Thus, $w\in L_D\iff \varphi(w)\in L_0$, or equivalently,
$\chi_{L_0}(\varphi(w))=\chi_{L_D}(w)$. This mapping is also injective, so we
reach a similar implication as \eqref{eqn:l_n-implies-l_0} by the same token,
which is
\begin{align}
  [\exists C'\in\C: &\Pr_{w\in E_\ell}(C'\text{ decides }L_0)>\eps]\nonumber\\
  &\quad \to[\exists C''\in\C: \Pr_{w\in E_{\ell}}(C''\text{ decides }L_D)>\eps]\label{eqn:l_0-implies-l_D},
\end{align}
in which $C''=C'\circ\varphi$ is of polynomial size, since $\varphi$ can be
computed in polynomial time (and therefore is also computable by a polynomial
size circuit).

Upon chaining \eqref{eqn:l_n-implies-l_0} and \eqref{eqn:l_0-implies-l_D},
followed by a contraposition, we get
\[
[\forall C\in\C: \Pr_{w\in E_\ell}(C=\chi_{L_D})\leq\eps]\to[\forall C\in\C: \Pr_{w\in E_\ell}(C=\chi_{\LN})\leq\eps],
\]
and by taking the likelihoods for the converse events with $\delta=1-\eps$,
\begin{align}
[\forall C\in\C: &\Pr_{w\in E_\ell}(C\neq \chi_{L_D})\geq\delta]\nonumber\\
&\quad \to[\forall C\in\C: \Pr_{w\in E_\ell}(C\neq \chi_{\LN})\geq\delta]\label{eqn:forall-implication},
\end{align}
using the notation $C=\chi$ and $C\neq\chi$ to mean that $C$ correctly or
incorrectly decides the respective language.

Thus, to prove that every circuit of polynomial size will incorrectly decide
$\LN$, and therefore incorrectly recover the first input bit $b_1$,
conditional on $E_\ell$, we need to lower-bound the likelihood for a
polynomial-size circuit to err on deciding $L_D$, and get rid of the
conditioning on $E_\ell$. Lemma \ref{lem:conditioning} helps with the latter,
as we get an $\ell_0>0$ so that for all $\ell>\ell_0$,

\begin{equation}\label{eqn:conditional-likelihood-lower-bound}
    \Pr_{w\in E_\ell}(C\neq\chi_{L_D})=\Pr_{w\in\Sigma^*}(C\neq\chi_{L_D}|E_\ell)\geq \frac 1 2\cdot \Pr_{w\in\Sigma^*}(C\neq\chi_{L_D})
\end{equation}

\begin{rem}
Two further intuitive reasons for the convergence of
$\Pr_{w\in\Sigma^*}(C\neq\chi_{L_D})\to
\Pr_{w\in\Sigma^*}(C\neq\chi_{L_D}|E_\ell)$ can be given: first, note that
our consideration of the decision on $L_D$ is focused on the first bit $b_1$,
while the event $E_\ell$ is determined by the other bits $b_n,
b_{n+1},\ldots$ of the input, where $n>1$. Since these are stochastically
independent of $b_1$, the related events are also independent. Second, the
selection algorithm is constructed to take elements disregarding their
particular inner structure, and hence independent of the condition $w\in
L_D$. Thus, the event of a correct selection ($E_\ell$) is independent of the
event $w\in L_D$.
\end{rem}

Because $C$ is by definition an acyclic graph, the computation of $C(w)$ can
be done by a \ac{TM} via evaluating all gates in the topological sort order
of (the graph-representation of) $C$. Moreover, it is easy to design a
universal such circuit interpreter \ac{TM} $M_{UC}$ taking a description of a
circuit $C$ and a word $w$ as input to compute $C(w)$ in time
$\poly(\size(C))$. In our case, since $C$ has $\size(C)=S(\ell)$, where $S$
is a polynomial, the simulation of $C$ by $M_{UC}$ takes polynomial time
$\geq S(\ell)$ again.

Remembering our notation from Section \ref{sec:tm-encoding}, we write $M_w$
for the \ac{TM} being represented by a word $w\in\Sigma^*$. Likewise, let us
write $M_C$ for the \ac{TM} that merely runs the universal circuit
interpreter machine $M_{UC}$ on the description of the circuit $C$. If, for
some word $w$ and circuit $C$, $M_w$ and $M_C$ compute the same function on
all $\Sigma^*$, we write $M_w\equiv_f M_C$ (to mean ``functional
equivalence'' of $M_w$ and $M_C$). With this notation, let the event
``$M_C\neq \chi$'' be defined identically to ``$C\neq\chi$''.

To quantify the right-hand side probability in
\eqref{eqn:conditional-likelihood-lower-bound}, let us return to the proof of
Theorem \ref{thm:time-hierarchy} again: the key insight is that the language
$L_D$ is defined to include all words $w$ for which the \ac{TM} $M_w$ would
reject ``itself'', i.e., $w$, as input, and has enough time to carry to
completion. Since the \ac{TM} $M_C$ that equivalently represents the circuit
$C$ above would accept its own string representation $w$ but $L_D$ is defined
to \emph{exclude} exactly this word, $M_C$ (and therefore also $C$) would
incorrectly compute the output for at least all words that represent
sufficiently large encodings of $M_C$. Formally,
\[
\Pr_{w\in\Sigma^*}(C\neq\chi_{L_D})\geq\frac{\abs{\set{w\in \Sigma^\ell: M_w\equiv_f M_C}}}{2^\ell}\stackrel{\eqref{eqn:equivalent-encodings}}{\geq}\frac 1 2\cdot\frac{2^{\ell-\log\ell}}{2^\ell}=\frac 1{2\ell},
\]
where we have used the (wasteful) encoding of \acp{TM} introduced in Section
\ref{sec:tm-encoding}. Plugging this into
\eqref{eqn:conditional-likelihood-lower-bound} tells us that
\begin{equation}\label{eqn:conditional-likelihood-lower-bound-2}
\Pr_{w\in\Sigma^*}(C\neq\chi_{L_D}|E_\ell)\geq\frac 1{4\ell},
\end{equation}
which is a universal bound that is independent of the particular circuit $C$.
So, let $C$ be arbitrary and of polynomial size $\leq S(\ell)$. We use
implication \eqref{eqn:forall-implication} with
\eqref{eqn:conditional-likelihood-lower-bound-2}, to conclude $\Pr_{w\in
E_\ell}(C\neq \chi_{\LN})\geq 1/(4\ell)$. The actual interest, however, is on
the \emph{unconditional} likelihood of $C$ outputting $b_1$ incorrectly. For
that matter, we invoke Lemma \ref{lem:conditioning} on
\eqref{eqn:c-error-bound}, to obtain a value $\ell_1>0$ so that for all
$\ell>\ell_1$,
\[
\Pr_{w\in\Sigma^\ell}[C\text{ incorrectly computes }b_1\text{ from }f_\ell(w)]\geq \frac 1 2\cdot\Pr_{w\in
E_\ell}(C\neq \chi_{\LN})\geq \frac 1{8\ell}.
\]
By taking the converse probabilities
again in \eqref{eqn:c-error-bound}, we end up with
\[
\Pr_{w\in\Sigma^\ell}[C(f_\ell(w))\in f_\ell^{-1}(f_\ell(w))]<1-\frac 1{8\ell},
\]
for all $\ell>\max\set{\ell_0,\ell_1}$ and every circuit $C$ of polynomial
size $S(\ell)$.

\section{Barriers towards an Answer about $\P$-vs-$\NP$}\label{sec:p-vs-np}

A purported implication (see, e.g., \cite{Zimand2004}) of Theorem
\ref{thm:weak-owf-exist} is the following separation:
\begin{cor}\label{cor:n-not-equal-to-np}
$\P\neq\NP$.
\end{cor}
Before attempting to prove Corollary \ref{cor:n-not-equal-to-np}, we first
ought to check if the results we have are admissible (able) to deliver the
ultimate conclusion claimed above. Our agenda in the following concerns three
``meta-conditions'' that can render certain arguments ineffective in proving
Corollary \ref{cor:n-not-equal-to-np}. The barriers are:
\begin{itemize}
  \item relativization \cite{Baker.1975},
  \item algebrization (a generalization of relativization)
      \cite{Aaronson2009}, and
  \item naturalization \cite{Razborov1997}.
\end{itemize}
There is also a positive (meta-)result pointing at a direction that any
successful proof of Corollary \ref{cor:n-not-equal-to-np} must come from,
which is \emph{local checkability} \cite{Arora.2007} (this describes an axiom
to which arguments for $\P\neq \NP$ must be consistent with). We need to
argue that the three barriers above are not in our way, but we also need to
show consistency with local checkability. It should be stressed that all of
these (four) conditions can only provide guidance towards taking the right
approach in proving Corollary \ref{cor:n-not-equal-to-np}. Our basic starting
point will be Theorem \ref{thm:weak-owf-exist}, but our objective is
\emph{not} on substantiating its truth (which should only be verified upon
correctness of all steps taken to concluding it), but to use the insights
cited above as a compass when arguing about $\P$-vs-$\NP$ based on Theorem
\ref{thm:weak-owf-exist}. Instead, we will exhibit the proof as a whole to
non-relativize, non-algebrize and non-naturalize by exhibiting one argument
in it that does not relativize, algebrize or naturalize\footnote{This is in
analogy to how non-naturalizing results were exposed as algebrizing, since
many of those had a sequence of all relativizing (and hence algebrizing)
arguments with only one non-relativizing argument that still algebrized (see
\cite{Aaronson2009}).}.


In general, the difficulty of proving $\P\neq\NP$ may root in one of three
possibilities, which are: (i) the claim is independent of \ac{ZFC}, in which
case, the separation would not be provable at all; (ii) it is wrong, which
would imply the yet unverified existence of polynomial-time algorithms for
every problem in $\NP$; or (iii) it is provable yet we have not found a
technique sufficiently powerful to accomplish the proof. The third
possibility has been studied most intensively, and also relates to proofs of
independence of $\P\neq\NP$ from \ac{ZFC}.

\subsection{Relativization}\label{sec:relativization}

In fact, under suitable models, i.e., assumptions made in the universe of
discourse, either outcome $\P=\NP$ and $\P\neq\NP$ is possible, so above all,
any argument that could settle the issue must not be robust against arbitrary
assumptions being made. This brings us to the concept of
\emph{relativization}. Formally, we call a complexity-theoretic statement
$\textsc{C}\subseteq\textsc{D}$ (resp. $\textsc{C}\not\subseteq\textsc{D}$)
\emph{relativizing}, if $\textsc{C}^A\subseteq\textsc{D}^A$ (resp.
$\textsc{C}^A\not\subseteq\textsc{D}^A$) holds for all oracles $A$. Here, the
oracle is the specific assumption being made, and it has been shown (using
diagonalization) that certain assumptions can make the claim $\P\neq\NP$
either true or false:
\begin{thm}[{Baker, Gill and Solovay \cite{Baker.1975}}]\label{thm:oracles-do-not-work-for-p-vs-np}
There are oracles $A$ and $B$, for which $\P^A=\NP^A$ and $\P^B\neq \NP^B$.
\end{thm}

If Theorem \ref{thm:weak-owf-exist} remains true in a universe that offers
oracle access to either $A$ or $B$, the conclusion thereof about
$\P$-vs-$\NP$ would -- in any outcome -- contradict Theorem
\ref{thm:oracles-do-not-work-for-p-vs-np}. More specifically, if Theorem
\ref{thm:weak-owf-exist} leads to $\P\neq\NP$ and the arguments used to this
end relativize, then the obvious inconsistency with Theorem
\ref{thm:oracles-do-not-work-for-p-vs-np} would imply that either Theorem
\ref{thm:weak-owf-exist} or its Corollary \ref{cor:n-not-equal-to-np} are
flawed.

Does the proof of Theorem \ref{thm:weak-owf-exist} relativize? The answer is
no, but not visibly so at first glance. Classifying an argument as
relativizing must consider the technical way of oracle access (e.g., whether
the space on the oracle tape counts towards the overall space complexity,
etc.). An excellent account for the issue is provided by L. Fortnow
\cite{Fortnow1994}, who discusses different forms of relativization. His work
eloquently exposes the issue as being strongly dependent on the mechanism
used to query the oracle. A usually non-relativizing technique is
\emph{arithmetization}, which transfers the operations of a circuit or a
\ac{TM} to a richer algebraic structure, typically a (finite) field $\F$,
where the armory to analyze and prove things is much stronger. A prominent
application and hence non-relativizing result is Shamir's theorem stating
that $\textsc{Ip}=\textsc{Pspace}$. However, by adapting the oracle query
mechanism suitably, even results obtained by arithmetization can relativize.
Specifically, Theorem 5.6 in \cite{Fortnow1994} is a version of Shamir's
theorem that \emph{does} relativize under the notion of an \emph{algebraic
oracle}. Subsequently, this concept was generalized and coined
\emph{algebrization} in \cite{Aaronson2009}, who exhibited a large number of
previously non-relativizing techniques to algebrize, so that proven
inclusions remain valid under this new kind of oracle power. Hereafter, we
will not confine ourselves to a particular method of oracle access, and
instead let the oracle only ``be available'' in either classical,
arithmetized or algebraic form. Since the classical oracle access by querying
some set $A$ is only generalized by subsequent findings, our argument will be
developed around the simplest form of oracles, stepwise showing how the
conclusions remain true in light of generalized forms of oracles.

Note that the proof of Theorem \ref{thm:weak-owf-exist} never speaks about
oracles or intractability, except during the diagonalization used to prove
the Time Hierarchy Theorem. A standard diagonalization argument does
relativize upon a syntactic change by letting all \acp{TM} be
oracle-\acp{TM}. However, the particular classes $\dtime(t(n))$ and
$\dtime(2^n)$ that we fixed in Section \ref{sec:weak-owf-exist} cannot be
separated (not even by diagonalization) in certain relativized worlds. In
fact, we can even derive an analogue result to Theorem
\ref{thm:oracles-do-not-work-for-p-vs-np} by showing different oracles under
which Theorem \ref{thm:weak-owf-exist} fails, resp. holds (although only its
failure is actually required here to dispel concerns about the relativization
barrier).

The diagonalization argument, made concrete by our choices of $t(n)$ and
$T(n)$ in Section \ref{sec:weak-owf-exist}, delivered the language
$L_D\in\dtime(2^n)\setminus\dtime(t(n))$, where $t(n)$ is defined by
\eqref{eqn:superpolynomial-t}. The two complexity classes are embedded inside
the chain
\begin{equation}\label{eqn:chain}
  \P\subsetneq\dtime(t(n))\subsetneq\dtime(2^n)\subsetneq\exptime
\end{equation}
Let $A$ be any \exptime-complete language, such as $A=\{(M,k):$~the \ac{TM}
$M$ halts within $k$ steps (where $k$ is given in binary\footnote{If $k$ were
in unary notation, $A$ would be(come) $\P$-complete.})$\}$ and use this
language $A$ as an oracle. Then $\exptime^A=\exptime\subseteq\P^A$, which
implies all equalities in \eqref{eqn:chain} and in particular
$\dtime(t(n))^A=\dtime(2^n)^A$. This, however, destroys the whole fundament
of the construction underlying Theorem \ref{thm:weak-owf-exist} (indeed, the
question defining $A$ is exactly what the diagonalization is about). In fact,
we can even extend the finding closer towards Theorem
\ref{thm:oracles-do-not-work-for-p-vs-np} by a simple modification: let us
use the hierarchy theorem to squeeze a complexity class $\textsc{C}$ in
between $\dtime(t(n))$ and $\dtime(2^n)$, so that $\dtime(t(n))\subsetneq
\textsc{C}\subsetneq \dtime(2^n)$ by virtue of a language $B\in
\textsc{C}\setminus \dtime(t(n))$. In using $B$ as an oracle, we see that
$\dtime(t(n))^{B}\subseteq
\textsc{C}^{B}=\textsc{C}\subsetneq\dtime(2^n)^{B}$. The two classes
underlying the \ac{OWF} construction thus remain separated under the oracle
$B$, but become equalized under the oracle $A$. Thus, our argument does not
relativize, or formally (for later reference):
\begin{lem}\label{lem:non-relativization}
There are decidable languages $A$, $B$ for which
$\dtime(t(n))^{A}=\dtime(2^n)^{A}$ and $\dtime(t(n))^{B}\neq
\dtime(2^n)^{B}$. Thus, Theorem \ref{thm:weak-owf-exist} fails in a world
relativized by $A$ and holds in the world relativized by $B$.
\end{lem}

\subsection{Local Checkability}\label{sec:local-checkability}

A condition that partly explains why proofs do not relativize is \emph{local
checkability} \cite{Arora.2007}. To formally define the concept and exhibit
Theorem \ref{thm:weak-owf-exist} as consistent with this framework, let us
briefly review the notion of a \emph{proof checker}: This is a \ac{TM} $M$
that uses universal quantification and an auxiliary input proof string $\Pi$
to accept an input string $x$ as being in $L$, if and only if all branches
(induced by the $\forall$ branching) accept. Otherwise, for $x\notin L$, the
machine $M$ should reject its input pair $(x,\Pi)$ for all $\Pi$. Herein, $M$
is allowed random access to $x$ and the proof string $\Pi$. The set of all
languages $L$ for which $M$ runs in time $\tau(n)$ is called the class
$\textsc{Pf-Chk}(\tau(n))$. A variation thereof is obtained by restricting
access to the proof string to only a subset of at most $\tau(n)$ bits, which
induces some sort of ``locality'' in the way the proof string can be used
(more technically, arbitrarily (e.g., exponentially) long queries to the
oracle can be precluded by the locality requirement). The resulting class is
called $\textsc{WPf-Chk}(\tau(n))$, and we refer to \cite{Arora.2007} for a
formal definition. For our purposes, it suffices to discuss the most
important implications of this concept:
\begin{enumerate}
  \item The \textbf{\ac{LCT}} \cite[Prop.4 and 5]{Arora.2007}:

      $\textsc{WPf-Chk}(\log n)=\NP=\textsc{Pf-Chk}(\log n)$, where the
      latter equality follows by inspecting the proof of the Cook-Levin
      theorem.
  \item For a random oracle $A$, we have
      $\P^A\not\subseteq\textsc{(W)Pf-Chk}(\tau(n))^A$ with probability 1,
      although $\P\subseteq \textsc{(W)Pf-Chk}(\tau(n))$ by the local
      checkability theorem. It follows that unrestricted (random) oracles
      appear unrealistic \cite{Arora.2007}, and therefore, we can restrict
      attention to oracles $A$ that are \emph{consistent with the
      \ac{LCT}}, which are those for which $\textsc{WPf-Chk}(\log
      n)^A=\NP^A$.
  \item Oracles being in that sense consistent with the \ac{LCT}, however,
      are allowed in proofs about $\P$ vs. $\NP$, since
      \cite[Thm.8]{Arora.2007}: If $\P^A\neq\NP^A$ for an oracle $A$ that
      is consistent with \ac{LCT}, then $\P\neq\NP$ (note that no analogous
      result holds for arbitrary, i.e, unrestricted, oracles).
\end{enumerate}

An objection against this concept as an explanation of the so-far observed
failure to prove $\P\neq\NP$ is the different style of oracle access used in
$\textsc{WPf-Chk}$ and $\NP$, which brings us back to the previous remarks
quoting \cite{Fortnow1994}. The concept of locality has been introduced in
\cite{Arora.2007} to partly address this issue, and is in fact enforced by
the encoding (see Figure \ref{fig:encoding}) that was used to make the
worst-case occur with the desired frequency. Indeed, the universal \ac{TM}
that we used here processes only a logarithmically small fraction of its
input, which corresponds to the $\log$ bound appearing in $\textsc{WPf-Chk}$
above (as we are simulating a \ac{TM} encoded by a word $w$ on input $w$, the
input pair to the proof checker would be $(x,\Pi)=(w,w)$, but the universal
\ac{TM} is constructed to use only $O(\log(\len{w}))$ bits of $\Pi=w$). So,
the proof of Theorem \ref{thm:weak-owf-exist} complies with the \ac{LCT}.

\subsection{Algebrization}\label{sec:algebrization}
Here, we let the oracle be a Boolean function $A_m:\set{0,1}^m\to\set{0,1}$
(instead of some general set). An extension of $A_m$ over some (finite) field
$\F$ is a polynomial $\tilde{A}_{m,\F}:\F^m\to\F$ such that
$\tilde{A}_{m,F}(x)=A_m(x)$ whenever $x\in\set{0,1}^m$. The oracles
considered for algebrization are the collections $A=\set{A_m:m\in\N}$ and
$\tilde{A}=\{\tilde{A}_{m,\F}:m\in\N\}$, and the algorithms are given oracle
access to $A$ or $\tilde{A}$. The inclusions of interest are separations like
$\textsc{C}\neq \textsc{D}$. Those are said to \emph{not algebrize}, if there
exist oracles $A,\tilde{A}$ such that $\textsc{C}^{\tilde{A}}=\textsc{D}^{A}$
(in an attempt to resemble the usual relativization taking the same oracles
on both sides, L. Fortnow \cite{Fortnow1994} used a much more complicated
construction of what he calls an algebraic oracle. The definition here is
from \cite{Aaronson2009} and designed to be more flexible and easier to use).

Let us reconsider $\dtime(t(n))\subsetneq\dtime(2^n)$: we recognized that
relation as non-relativizing due to the oracle language $A$ that equalized
the two classes. The point for now is that $A$ is a \emph{decidable}
language, so that there is a \ac{TM} to compute $\chi_A$. This \ac{TM} can be
converted into a circuit family $\set{A_m:m\in\N}$ with help of the
Pippenger-Fisher theorem \cite{Pippenger&Fischer1979} (and arithmetized in
the usual way). Queries to the (set) $A$ can thus be emulated by calling the
function $A_m$ to compute the indicator function $\chi_A$ for inputs of size
$m$. The query size (left unrestricted in the plain definition of
algebrization) to the oracle is (due to our encoding) also bound to be
logarithmic (as noted before). For retaining the result of Lemma
\ref{lem:non-relativization}, we can put $A=\tilde{A}$ (as a trivial
extension), and the identity $\dtime(t(n))^A=\dtime(2^n)^A$ is implied in the
so-algebrized world. Thus, the separation in which Theorem
\ref{thm:weak-owf-exist} roots does not algebrize either.

\subsection{A Formal Logical View on (Algebraic) Relativization}
Though the argument underlying Theorem \ref{thm:weak-owf-exist} is by the
above token not relativizing \emph{in general}, the real point of
relativization and algebrization is deeper: Since the existence of \ac{OWF}
would point towards $\P\neq\NP$, the question is whether a proof of this
claim relativizes or algebrizes with oracles that equalize $\P$ and $\NP$.
Moreover, the most interesting oracles for that matter would be outside
$\P$~\footnote{Otherwise, if the oracle is in $\P$, then the oracle mechanism
of any \ac{TM} $M^A$ could be integrated into the logic of the machine to
deliver an equivalent \ac{TM} that behaves exactly as $M^A$, but uses no
oracle at all; hence is tantamount to unconditionally assuming $\P=\NP$ from
the beginning.}. To compactify the discussion hereafter, the term ``oracle''
will synonymously mean both, sets (as in Section \ref{sec:relativization})
and algebraic oracles being Boolean functions (as in Section
\ref{sec:algebrization}).

Let us take a look at relativization and algebrization from a perspective of
formal logic: let $A$ be the logical statement that an oracle is available
(in the form of an oracle \ac{TM} or oracle circuit), and let $PROOF$ be the
conjunction of arguments towards a claimed relation between $\P$ and $\NP$. A
relativizing or algebrizing proof is one for which $A\land PROOF$ is
consistent in the sense of being logically true under the chosen
interpretation and universe of discourse (where $PROOF$ is syntactically
modified to use assumption $A$ wherever this is appropriate). Suppose that
this implies a contradiction (say, an inconsistency with Theorem
\ref{thm:oracles-do-not-work-for-p-vs-np} or with the results in
\cite{Aaronson2009}), then, based on this contradiction, the common
conclusion is that $PROOF$ must be wrong, since ``the proof is
relativizing/algebrizing''.

The claim made here is that this final conclusion can be flawed, as it misses
the fact that the (proven) \emph{existence} of an oracle in general
\emph{does not} imply the existence of a mechanism to query it! For example,
if the oracle is an undecidable language, despite its verified existence, no
oracle-\ac{TM} $M^A$ can (practically) exist; simply because no $M$ could
ever query $A$. Likewise, as another example, if the oracle is some
\NP-complete language, we merely \emph{assume} -- without verification or
proof -- that the problem $A$ can be solved in some unspecified way, which
implies that $\P$ and $\NP$ would be equal (as an a-priori hypothesis, this
is obviously inconsistent with the separation of the two classes that is
supposed to follow from the existence of \ac{OWF}; Theorem
\ref{thm:weak-owf-exist}). Thus, let $A$ be an(y) oracle, against which
$PROOF$ shall be tested to (not) relativize. The full assumption made along
such arguments is actually \emph{twofold}, since it concerns (i) the
existence of the oracle, and (ii) also the ability to query it, i.e., the
existence of the respective oracle-\ac{TM}. The first partial assumption (i)
is typically verified, but despite the significance of the query mechanism
(as eloquently pointed out by \cite{Fortnow1994} and demonstrated by the
whole idea of algebrization), the second implicit assumption (ii) is often
left unverified. Thus, the rejection of $PROOF$ because $A\land PROOF$ is
contradictive (i.e., $PROOF$ relativizes/algebrizes), rests on the
\emph{unverified hypothesis} that the oracle algorithm using $A$ actually
exists; for this to hold, however, the existence of the oracle alone is
insufficient in general (as follows from the above examples).

Lacking a proof of existence for the oracle \emph{and} the respective oracle
query mechanism, we are left with at least two possible (not mutually
exclusive) answers to as why $A\land PROOF$ yields a contradiction: 1) $A$ is
wrong, i.e., the oracle cannot be reasonably assumed available for queries
(though it may provably exist), or 2) $PROOF$ is wrong, i.e., the arguments
in the proof are flawed at some point. Finding out which of the two possible
answers is correct requires either a proof that $A$ is true, meaning that
oracle queries can practically work as assumed (this is a usually undiscussed
matter in the literature), or inspecting $PROOF$ for logical consistency and
correctness (as is the standard procedure for all mathematical proofs
anyway)\footnote{The inherent symmetry can be taken further: If $PROOF$ is
verifiably true based on a pure judgement of arguments, and its relativized
version leads to a verified contradiction, then the oracle hypothesis $A$
must be wrong. If the oracle itself is existing (again, provably), then the
only possible remaining conclusion is that the query mechanism must be
impossible. So, relativization can even be a method to prove the
\emph{practical non-existence} of certain oracle-algorithms; a possibility
whose exploration may be of independent interest.}.

It follows that relativization and algebrization are effective barriers only
if the oracle under which the inconsistency with the argument in question
arises, exists and is provably useable in the sense as the oracle query
mechanism assumes it. Otherwise, the finding in the respective relativized
world remains in any case conditional on the oracle hypothesis\footnote{The
choice of the oracle as such is crucial already, as a random choice of the
oracle is known to be a dead end in this context \cite{Chang.1994}.}, and we
cannot reliably tell which is wrong: the hypothesis or the proof arguments?
The insight that not all oracles are equally useful to reason about how $\P$
relates to $\NP$ is actually not new, as local checkability
(\cite{Arora.2007}; Section \ref{sec:local-checkability}) is an independent
earlier discovery in recognition of similar issues.

Irrespectively of the above, it is possible to modify the proof of Theorem
\ref{thm:weak-owf-exist} so that it deteriorates in worlds where oracles come
into play. The idea is to explicitly account for any use of the oracle in the
definition of the diagonal language $L_D$. Recall the overall construction in
the proof of Theorem \ref{thm:time-hierarchy}, which Figure \ref{fig:dtht}
depicts. Call the output decision $d\in\set{0,1}$ and add the following
logical condition to the way how this construction defines $L_D$:

\begin{equation}\label{eqn:no-oracle-condition}
\begin{array}{l}
\text{\textbf{if} the oracle $A$ was called during the simulation of $M_w$}\\
\text{\textbf{then return} $(1-d)$ \textbf{else return} $d$.}
\end{array}
\end{equation}

\begin{figure}
  \centering
  \includegraphics[width=\textwidth]{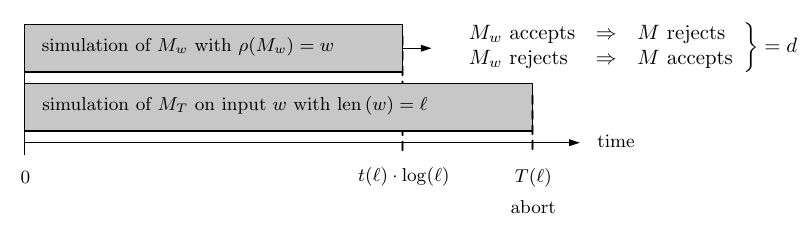}
  \caption{Simulation Setup for the Time Hierarchy Theorem}\label{fig:dtht}
\end{figure}

This changes \eqref{eqn:diagonal-language} by rephrasing $L_D$ into
containing all words for which either of the following two conditions hold:
\begin{enumerate}	
\item the simulation of $M_w$ halts and rejects $w$ within $\leq
    T(\len{w})$ steps, provided that $M_w$ makes no call to any oracle
    (i.e., acts as in a non-relativized world),
\item the simulation of $M_w^A$ halts and \emph{accepts} $w$ within $\leq
    T(\len{w})$ steps (now, there was a call to the oracle, so that the
    simulation was done for $M_w^A$ necessarily and the upper additional
    condition hence inverted the rejection into an acceptance behavior).
\end{enumerate}

Thus, upon relativization using the oracle $A$, the final argument towards
proving Theorem \ref{thm:time-hierarchy} deteriorate into a humble tautology:
$w\in L(M_w^A)$ implies $w\in L_D$ and vice versa, so the contradiction that
separates $\dtime(t)$ from $\dtime(T)$ can no longer be reached. Observe that
this \emph{does not} mean that the classes are not separated for other
reasons, but this particular argument no longer supports that claim. This
already suffices to escape relativization, since the so-modified reasoning
towards the statement of Theorem \ref{thm:weak-owf-exist} becomes void in any
relativized world where the oracle is actually used. In a non-relativized
world, however, there cannot be any call to any oracle, so that condition
\eqref{eqn:no-oracle-condition} has no effect whatsoever, and the first of
the two above cases will be the only one to apply. Thus, Theorem
\ref{thm:time-hierarchy} remains to hold and all arguments based upon go
unchanged.

What happens in worlds relativized by oracles that \emph{separate} $\P$ from
$\NP$? The argument breaks down in exactly the same way as before, and (also
as before) we can say nothing about the relation of $\P$ and $\NP$ then, so
no inconsistency arises here either.

These arguments remain intact also for algebrization, if condition
\eqref{eqn:no-oracle-condition} is rephrased into speaking about a perhaps
necessary ``evaluation'' of the oracle function. Circuits that lazy-evaluate
their logic thus may or may not need their oracle, so that the above
condition can be added to our proof with the same semantic and effect as
before. Thus, Theorem \ref{thm:weak-owf-exist}'s fundament will generally
collapse in algebrized worlds as well.


\subsection{Naturalization}
Regarding natural proofs, we may ask if a proof of $\P\neq\NP$ based on
Theorem \ref{thm:weak-owf-exist} is natural? The answer is (again) no! The
crucial finding of \cite{Razborov1997} is that any natural argument lends
itself to breaking pseudorandom generators. But in that case, we would also
get fast algorithms for some of the very same problems that we wanted to
prove hard by showing that $\P\neq\NP$ \cite{Aaronson2009}. This is the
barrier that natural proofs constitute, but starting from Theorem
\ref{thm:weak-owf-exist} lets us bypass this obstacle.

The reason why a natural property $C^*$ can be used to break a pseudorandom
generator is the disjointness of $C^*$ with the image set of some
pseudorandom function (constructed from the \ac{PRNG} in a similar style as
in \cite{Goldreich&Goldwasser&Micali1986}), provides a statistical test to
distinguish random from pseudorandom (output ensembles). That test employs
the poly-time decidability of $C^*$ (provided since $C^*$ is natural
\cite{Razborov1997}). Weak \ac{OWF} exist if and only if strong \ac{OWF}
exist \cite{Goldreich2003b}, so Theorem \ref{thm:weak-owf-exist} indirectly
gives a strong \ac{OWF} (see \cite{Zimand2004} or
\cite[Thm.2.3.2]{Goldreich2003b}), which in turn let us construct \ac{PRNG}
whose output cannot be distinguished from uniformly random in polynomial time
(see \cite[Def.3.3.1, Thm.3.5.12]{Goldreich2003b}). This contradiction rules
out any statistical test, including the aforementioned one based on deciding
$C^*$. Hence, in light of Theorem \ref{thm:weak-owf-exist}, a natural
property $C^*$ cannot exist at all (as is also explicitly said in
\cite[pg.3]{Razborov1997}). So, the proof of Corollary
\ref{cor:n-not-equal-to-np} based on Theorem \ref{thm:weak-owf-exist} is not
natural\footnote{Naturalization has (until today) nothing to say about proofs
regarding the existence of one-way functions (in the form used here; not
speaking about the entirety of all kinds of \acp{OWF}, since their existence
is currently \emph{not} known to follow from $\P\neq\NP$). An independent
concrete indication towards the proof of Theorem \ref{thm:weak-owf-exist} to
be non-natural is its use of diagonalization, which is typically considered
as a non-naturalizing argument \cite{Aaronson2009}.}.

\subsection{On the Separation of $\P$ from $\NP$}
It appears anyway questionable whether we are interested in answering $\P$
vs. $\NP$ in \emph{all possible} worlds, rather than under the more realistic
assumption of having no particular magic at hand (in the form of an oracle).
After all, the question is whether $\P$ is equal (or not) to $\NP$, given
those (and only those) operations that Turing machines can do. Note that our
use of the hierarchy theorem must not be mistakenly interpreted as the
high-level claim that $\P=\NP$ would contradict the hierarchy theorem. The
inclusion that we use relates to classes beyond $\P$ and hence also above
$\NP$ under the assumption $\P=\NP$. So, the hierarchy theorem remains an
unshaken base.

Taking Theorem \ref{thm:weak-owf-exist} as a fundament, we can now complete
our discussion by providing the proof of Corollary
\ref{cor:n-not-equal-to-np} in full detail.

\begin{proof}[Proof of Corollary \ref{cor:n-not-equal-to-np}]
Let $f:\Sigma^*\to\Sigma^*$ be a strong one-way function, whose existence is
implied by that of weak one-way functions by \cite[Thm.5.2.1]{Zimand2004}.
W.l.o.g., we may assume $\Sigma=\set{0,1}$ (otherwise, we just use a
prefix-free binary encoding to represent all symbols in the finite alphabet
$\Sigma$). Moreover, let $gn:\Sigma^*\to\N$ be a G\"odel numbering, for which
$gn(w)$ and $gn^{-1}(n)$ are both computable in polynomial time in $\len{w}$
and $\log(n)$, respectively. Our choice here is the function $gn$ from
Section \ref{sec:goedel-and-density}. We put $g:\N\to\N$ as $g:=gn\circ
f\circ gn^{-1}$, and observe that by \eqref{eqn:goedelnr}, $g$ inherits the
length regularity property from $f$ (where the integer $n$ has a length
$\len{n}\in \Theta(\log n)$ equal to the number of bits needed to represent
it). Furthermore, $g$ is as well strongly one-way: if it were not, i.e., if
$g^{-1}(n)$ would be computable in time $\poly(\log n)$, then
\begin{equation}\label{eqn:inverse-f}
f^{-1} = gn^{-1}\circ g^{-1}\circ gn
\end{equation}
would also be computable in time $\poly(\len{w})$ (since $gn$ and $gn^{-1}$
are both efficiently computable). Precisely, if some circuit $C$ of
$\size(C)\leq \poly(\log n)$ computes $C(n)\in g^{-1}(n)$ with a likelihood
of $\geq 1/\poly(\log n)$, then each of these cases is ``positive'' for the
computation of $f^{-1}$ on the entirety of the function's domain
$\set{0,1}^{\ell}$ with $\ell\in\Theta(\log n)$ (where the $\Theta$ is due to
the application of $gn$ and $gn^{-1}$). This means that the circuit $C$ could
be extended into a (polynomial size) circuit $C'$ that evaluates $f^{-1}$
according to \eqref{eqn:inverse-f} correctly with a likelihood $\geq
1/\poly(\Theta(\ell))$, contradicting the strong one-wayness of $f$.

Upon $g$, we define a language
\[
L_g:=\set{(y,N)\in\N^2: \exists
x\in\set{1,\ldots, N}\text{~with~}g(x)=y},
\]
in which every pair $(y,N)$ can be represented by a word
$w\in\set{0,1}^{\Theta(\log y+\log N)}\in\Sigma^*$ using a proper prefix-free
encoding (which includes the symbols to separate the binary strings for $y$
and $N$). That is, $L_g$ is the set of $y$ for which a preimage within a
specified (numeric) range $[1,N]$ exists. Our goal is showing that
$L_g\in\NP$ but $L_g\notin \P$.

The observation that $L_g\in\NP$ is immediate, since a preimage $x$ for
$y\in\N$ has length $O(\log x)$, so it can act as a polynomial witness,
guessed by a nondeterministic \ac{TM} to decide $1\leq x\leq N$ and $g(x)=y$,
both doable in time $O(\poly(\log x))$ (as $g$ is length-regular and strongly
one-way).

Conversely, if we assume $L_g\in\P$, then we could efficiently compute
$x=g^{-1}(y)$ for every given $y\in\N$ by the following method: since $g$ is
length regular, it satisfies $\len{y}=\len{g(x)}\geq \len{x}^{1/k}$, where
$\len{x}\in O(\log x)$ when $x$ is treated as a word in binary
representation. The value $k$ is a constant that only depends on $g$. Thus,
we have the upper bound $\log(x)\in O((\log y)^k)$, and therefore $x$ lies
inside the discrete interval $I=\{1,2,\ldots,N=c\cdot\lceil2^{(\log
y)^k}\rceil\}$ for some constant $c>0$ and sufficiently large $x$ (implied by
a sufficiently large $y$ via the length-regularity of $g$). With the
so-computed $N$, we run a binary search on $I$: per iteration, we can invoke
the polynomial-time decision algorithm $A$ available for $L_g\in\P$ to decide
whether to take the left half (if $A$ returns "yes") or the right half (if
$A$ returns "no") of the current search space. After $O(\log N)=O((\log
y)^k)$ iterations, the interval has been narrowed down to contain a single
number $x_0$, which is the sought preimage of $y$. The whole procedure takes
$O((\log y)^k)\cdot\poly(\log y)$ steps (one decision of $L_g$ per iteration
of the binary search), and thus is polynomial in $\log y$ since $k$ is a
constant. Therefore, $g^{-1}$ would be computable in $O(\poly(\log y))$ steps
in the worst case. Since our choice of $y$ was arbitrary, it follows that an
evaluation of $g^{-1}$ takes $O(\poly(\log y))$ steps in \emph{all} cases,
which clearly contradicts the average-case hardness of the strong one-way
function $g$. Hence, $L_g\notin\P$, and $\P\neq\NP$ consequently.
\end{proof}

\section*{Acknowledgment}
The author is indebted to Max-Julian Jakobitsch, Stefan Haan and Moritz Hiebler, for their hard work on formalization, resulting in the identification of subtle errors that were corrected thanks to their results. A likewise thank goes to Sandra K\"onig from the Austrian Institute of
Technology, for spotting some errors in earlier versions of the manuscript, as well as to Patrick Horster, for
valuable discussions about earlier versions of this manuscript.

\bibliographystyle{plain}

%
%
%
%
%
%
%
%

\begin{appendix}

\begin{acronym}
\acro{OWF}{one-way function}%
\acro{TM}{Turing machine}%
\acro{DTHT}{deterministic time hierarchy theorem}%
\acro{LCT}{local checkability theorem}%
\acro{PRNG}{pseudorandom number generator}%
\acro{ZFC}{Zermelo-Fraenkel set theory with the axiom of choice}
\end{acronym}

\end{appendix}

\end{document}